\def\VersionFinal{}
\def\VersionLNCS{}
\def\VersionArXiv{}
\documentclass[a4paper,10pt,envcountsame]{llncs}
\usepackage{etoolbox}
\ifdefined\VersionArXiv%
	\undef\VersionWithComments{}
    \def\VersionLong{}
	\undef\VersionAnonymous{}
    \def\VersionFinal{}
\fi

\ifdefined\VersionLong%
	\newcommand{\LongVersion}[1]{#1}
	\newcommand{\ShortVersion}[1]{}
\else
	\newcommand{\LongVersion}[1]{}
	\newcommand{\ShortVersion}[1]{#1}
\fi

\ifdefined\VersionLNCS%
	\newcommand{\LNCSVersion}[1]{#1}
	\newcommand{\IEEEVersion}[1]{}
\else
	\newcommand{\LNCSVersion}[1]{}
	\newcommand{\IEEEVersion}[1]{#1}
\fi

\usepackage{comment}
\ifdefined\VersionLong%
        \includecomment{LongVersionBlock}
        \excludecomment{ShortVersionBlock}
\else
        \excludecomment{LongVersionBlock}
        \includecomment{ShortVersionBlock}
\fi
\excludecomment{lipics}
\includecomment{lncs}
\excludecomment{ieee}
\excludecomment{NoSecurityVersion}

\usepackage[ruled,lined,linesnumbered,noend]{algorithm2e}
\SetKwInOut{Input}{Input}
\SetKwInOut{Output}{Output}
\SetKw{KwRemove}{remove}
\SetKw{KwPush}{add}
\SetKw{KwPop}{pop}
\SetKw{KwFrom}{from}
\SetKw{KwTo}{to}
\SetKw{KwCompute}{compute}
\SetKw{KwReturn}{return}
\SetKw{KwBreak}{break}
\SetKw{KwPick}{pick}
\SetKw{KwLet}{let}
\SetKw{KwBe}{be}
\SetKw{KwSplit}{split}
\SetKw{KwInto}{into}
\SetKw{KwFind}{find}
\SetKw{KwFeed}{feed}
\SetKwProg{Fn}{Function}{:}{}
\SetAlgorithmName{Algorithm}{algorithm}{List of Algorithms}

\begin{lncs}
 \usepackage{subcaption}
\end{lncs}
\begin{ieee}
 \usepackage{subcaption}
\end{ieee}

\begin{lncs}
 \usepackage{paralist} %
\end{lncs}
\begin{ieee}
 \usepackage{paralist} %
\end{ieee}

\begin{ieee}
 \usepackage[numbers]{natbib}
\end{ieee}

\usepackage{xspace}
\usepackage{adjustbox}
\usepackage{array}

\begin{lncs}
 \newenvironment{ienumeration}
	{\ifdefined\VersionLong\begin{enumerate}\else\begin{inparaenum}[\itshape i\upshape)]\fi}
	{\ifdefined\VersionLong\end{enumerate}\else\end{inparaenum}\fi}

 \newenvironment{oneenumeration}
	{\ifdefined\VersionLong\begin{enumerate}\else\begin{inparaenum}[1)]\fi}
	{\ifdefined\VersionLong\end{enumerate}\else\end{inparaenum}\fi}

 \newenvironment{myitemize}
	{\ifdefined\VersionLong\begin{itemize}\else\begin{inparaitem}[]\fi}
	{\ifdefined\VersionLong\end{itemize}\else\end{inparaitem}\fi}
\end{lncs}
\begin{ieee}
 \newenvironment{ienumeration}
	{\ifdefined\VersionLong\begin{enumerate}\else\begin{inparaenum}[\itshape i\upshape)]\fi}
	{\ifdefined\VersionLong\end{enumerate}\else\end{inparaenum}\fi}

\end{ieee}

\begin{ieee}
 \usepackage{amsthm}
\end{ieee}
\usepackage{amsmath} %
\usepackage{amssymb} %
\usepackage{stmaryrd} %
\usepackage{siunitx}
\usepackage{mathdots}
\usepackage{centernot}
\usepackage{multirow}
\usepackage{mathtools}

\usepackage[misc,geometry]{ifsym} %

\ifdefined\VersionWithComments%
	\usepackage{draftwatermark}
	\SetWatermarkText{draft}
 	\SetWatermarkScale{15}
	\SetWatermarkColor[gray]{0.9}
\fi
\usepackage[svgnames,table]{xcolor}
\usepackage{colortbl}
\definecolor{darkblue}{rgb}{0.0,0.0,0.6}
\definecolor{darkgreen}{rgb}{0, 0.5, 0}
\definecolor{darkpurple}{rgb}{0.7, 0, 0.7}
\definecolor{darkblue}{rgb}{0, 0, 0.7}

\usepackage[
		colorlinks=true,
		citecolor=darkgreen,
		linkcolor=darkblue,
		urlcolor=darkpurple,
	]{hyperref}

\usepackage[capitalise,english,nameinlink]{cleveref} %
\crefalias{AlgoLine}{line}
\crefname{line}{\text{line}}{\text{lines}} %
\crefname{item}{\text{item}}{\text{items}} %
\crefname{example}{\text{Example}}{\text{Examples}} %
\crefname{assumption}{\text{Assumption}}{\text{Assumptions}} %
\crefname{algorithm}{\text{Algorithm}}{\text{Algorithms}}

\usepackage{tikz}
\usetikzlibrary{arrows,automata,positioning,math,shapes,patterns,calc}
\tikzstyle{every node}=[initial text=]
\tikzstyle{location}=[rectangle, rounded corners, minimum size=12pt, draw=black, fill=blue!10, inner sep=2pt]
\tikzstyle{invariant}=[draw=black, dotted, inner sep=1pt] %
\tikzstyle{final}=[double]
\tikzstyle{accepting}=[final]
\tikzstyle{PTPMOPT}=[,dashed,color=red,semithick]
\usepackage{pgf-umlsd}

\ifdefined\VersionAnonymous%
\definecolor{coloract}{rgb}{0.0, 0.0, 0.0}
\else
\definecolor{coloract}{rgb}{0.50, 0.70, 0.30}
\fi
\definecolor{colorclock}{rgb}{0.4, 0.4, 1}
\definecolor{colorconst}{rgb}{0.50, 0.20, 0.00}
\definecolor{colordisc}{rgb}{1, 0, 1}
\definecolor{colorloc}{rgb}{0.4, 0.4, 0.65}
\definecolor{colorparam}{rgb}{1, 0.6, 0.0}

\tikzstyle{rqanswer} = [
 draw=black,
 fill=gray!30,
 text=black,
 line width=0.5pt,
  text width = \linewidth - 1.6 ex - 1pt,
  inner sep = 0.8 ex,
  rounded corners=4pt]
\newcommand{\rqanswer}[2]{\LongVersion{\smallskip}\noindent%
\smallskip
\begin{tikzpicture}%
\draw node[rqanswer]{\textbf{Answer to {#1}}:{ #2}};%
\end{tikzpicture}}

\usepackage{pgfplotstable}
\pgfplotsset{compat=1.12}
\usepackage{booktabs}

\usetikzlibrary{automata,positioning,matrix,shapes.callouts}
\tikzset{
accepting/.style={double distance=1pt}
}

\newcommand{\N}{{\mathbb{N}}}
\newcommand{\Z}{{\mathbb{Z}}}
\newcommand{\Zp}{{\mathbb{N}_{>0}}}
\newcommand{\R}{{\mathbb{R}}}

\newcommand{\Rp}{{\mathbb{R}_{>0}}}

\newcommand{\ttrue}{\mathrm{t{\kern-1.5pt}t}}
\newcommand{\ffalse}{\mathrm{f{\kern-1.5pt}f}}

\newcommand{\powerset}[1]{\mathcal{P} ({#1})}

\newcommand{\imply}{\Rightarrow}

\makeatletter
\newcommand{\figcaption}[1]{\def\@captype{figure}\caption{#1}}
\newcommand{\tblcaption}[1]{\def\@captype{table}\caption{#1}}
\makeatother

\usepackage{arydshln}
\tikzstyle{defproblem} = [
 draw=black,
 fill=cyan!10,
 text=black,
 line width=0.5pt,
  text width = \linewidth - 1.6 ex - 1pt,
  inner sep = 0.8 ex,
  rounded corners=4pt]
\newcommand{\defProblem}[3]
{%
\smallskip
\noindent\begin{tikzpicture}%
\draw node[defproblem]{%
\textbf{#1 problem:}\\%
\textsc{Input}: #2\\%
\textsc{Output}: #3};%
\end{tikzpicture}}%

\begin{ieee}
\theoremstyle{plain}
\newtheorem{theorem}{Theorem}
\newtheorem{lemma}[theorem]{Lemma}

\theoremstyle{definition}
\newtheorem{definition}[theorem]{Definition}
\newtheorem{example}[theorem]{Example}

\theoremstyle{remark}

\end{ieee}

\ifdefined\VersionWithComments%
	\usepackage[colorinlistoftodos,textsize=footnotesize]{todonotes}
\else
	\usepackage[disable]{todonotes}
\fi

\newcommand{\TODO}[1]{\todo{#1}}
\newcommand{\gennote}[3]{\todo[linecolor=#2,backgroundcolor=#2!25,bordercolor=#2]{#3: #1}}
\newcommand{\MW}[1]{\gennote{#1}{orange}{MW}}

\newcommand{\instructions}[1]{{\gennote{\bfseries #1}{red}{Instructions}}}
\newcommand{\reviewer}[2]{{\gennote{``#2''}{purple}{Reviewer #1}}}

\begin{lncs}
\ifdefined\VersionFinal%
\else
	\usepackage[pagewise]{lineno} %
	\linenumbers%
	
\fi
\end{lncs}
\begin{ieee}
\ifdefined\VersionFinal%
\else
	\usepackage[switch]{lineno} %
	\linenumbers%
	
\fi
\end{ieee}

\newcommand{\ourTool}{\textsc{ArithHomFA}}
\newcommand{\homfa}{\textsc{HomFA}}
\newcommand{\Naive}{${\ourTool}_{\textsc{Naive}}$}
\newcommand{\Optimized}{${\ourTool}_{\textsc{Opt}}$}
\newcommand{\BloodGlucose}{\textsf{BGLvl}}
\newcommand{\RSS}{\textsf{RSS}}

\newcommand{\goodFeature}[1]{\color{darkgreen}#1}
\newcommand{\badFeature}[1]{\color{red}#1}

\newcommand{\real}{r} %
\newcommand{\var}{x}
\newcommand{\Var}{\mathcal{V}}
\newcommand{\signal}{\sigma}
\newcommand{\valuation}{\sigma}
\newcommand{\prefix}[2]{#1_{\leq {#2}}}
\newcommand{\signalInfiniteInside}{\valuation_1,\valuation_2,\dots}
\newcommand{\signalFiniteInside}{\valuation_1,\valuation_2,\dots,\valuation_n}
\newcommand{\signalWithFiniteInside}{\signal = \signalFiniteInside}
\newcommand{\signalWithInfiniteInside}{\signal = \signalInfiniteInside}
\newcommand{\fml}{\varphi}
\newcommand{\fun}{\mu}
\newcommand{\AP}{\mathbf{AP}}
\newcommand{\satisfy}[3]{(#1,#2) \models #3}

\newcommand{\UntilOp}[1]{\mathbin{\mathcal{U}_{#1}}}
\newcommand{\Release}[1]{\mathbin{\mathcal{R}_{#1}}}

\newcommand{\TRelease}[1]{\mathbin{\overline{\mathcal{R}}_{#1}}}
\newcommand{\DiaOp}[1]{\Diamond_{#1}}
\newcommand{\BoxOp}[1]{\square_{#1}}

\newcommand{\NextOp}{\mathcal{X}}

\newcommand{\constant}{d}
\newcommand{\Pred}{\mathrm{Pred}}
\newcommand{\Monitor}{\mathcal{M}}
\newcommand{\plain}{p}

\newcommand{\scale}{\mathit{scale}}
\newcommand{\secretKey}{\mathbf{sk}}
\newcommand{\key}{\secretKey}
\newcommand{\publicKey}{\mathbf{pk}}
\newcommand{\cipherPair}[1][]{(b#1,\mathbf{a}#1)}
\newcommand{\RLWE}[1]{\overline{\mathbf{#1}}}
\newcommand{\RLWEcipherPair}[1][]{(\RLWE{b}#1,\RLWE{a}#1)}
\newcommand{\RGSW}[1]{\tilde{\mathit{#1}}}
\newcommand{\noise}{e}
\newcommand{\cipher}{\mathbf{c}}
\newcommand{\enc}{\mathbf{enc}}
\newcommand{\dec}{\mathbf{dec}}
\newcommand{\encCKKS}{\enc_{\mathrm{CKKS}}}
\newcommand{\decCKKS}{\dec_{\mathrm{CKKS}}}
\newcommand{\encTFHE}{\enc_{\mathrm{TFHE}}}
\newcommand{\decTFHE}{\dec_{\mathrm{TFHE}}}
\newcommand{\quotient}{q}
\newcommand{\qRing}[1][\quotient]{\ensuremath{\Z / #1\Z}}
\newcommand{\quotientCKKS}{\quotient_{\mathrm{CKKS}}}
\newcommand{\quotientTFHE}{\quotient_{\mathrm{TFHE}}}

\newcommand{\word}[1][]{w#1}
\newcommand{\action}{a}

\newcommand{\wordInside}[1][]{\action#1_1 \action#1_2 \dots \action#1_{n#1}}
\newcommand{\wordWithInside}[1][]{\word[#1]=\wordInside[#1]}
\newcommand{\BitEncoding}{\mathrm{BitEnc}}
\newcommand{\Lg}{\mathcal{L}}

\newcommand{\emptyword}{\varepsilon}

\newcommand{\BLOCK}{\textsc{Block}}
\newcommand{\REVERSE}{\textsc{Reverse}}
\newcommand{\reverse}[1]{{#1}^{\mathrm{rev}}}
\newcommand{\signalLength}{n}
\newcommand{\encSignal}{\RLWE{\cipher}^{\signal}}
\newcommand{\encSignali}[1]{\encSignal_{#1}}
\newcommand{\encSignalSequence}{\encSignali{1},\encSignali{2},\cdots,\encSignali{|\Var| \times \signalLength}}
\newcommand{\encWord}{\RGSW{\cipher}^{\word}}
\newcommand{\encWordSequence}{\encWord_{1},\encWord_{2},\cdots,\encWord_{\signalLength}}
\newcommand{\encResult}{\cipher^{\mathrm{res}}}
\newcommand{\encResultSequence}{\encResult_{1},\encResult_{2},\cdots,\encResult_{\signalLength}}
\newcommand{\BlockSize}{B}
\newcommand{\BootInterval}{I_{\mathrm{boot}}}
\newcommand{\diffCKKS}{\RLWE{\cipher}_{\quotientCKKS}}
\newcommand{\diffScaledCKKS}{\RLWE{\cipher}^{\mathit{scaled}}_{\quotientCKKS}}
\newcommand{\diffSwitched}[1][]{#1{\cipher}^{\mathit{scaled}}_{\quotient_{\mathrm{tmp}}}}

\newcommand{\signTFHE}[1][\RGSW]{#1{\cipher}^{\mathit{sign}}_{\quotientTFHE}}
\newcommand{\diffMax}{\mathit{maxDiff}}

\newcommand{\Loc}{L}

\newcommand{\A}{\mathcal{A}}

\newcommand\egoIndex{\mathrm{ego}}
\newcommand\precIndex{\mathrm{prec}}
\newcommand\dEgoPreBr{d_{\egoIndex}^{\mathit{preBr}}}
\newcommand\dEgoBrake{d_{\egoIndex}^{\mathit{Br}}}
\newcommand\dPrecBrake{d_{\precIndex}^{\mathit{Br}}}
\newcommand\dRSS{d_{\mathit{RSS}}}
\newcommand\xOf[1]{x_{#1}}
\newcommand\yOf[1]{y_{#1}}
\newcommand\vyOf[1]{\dot{y}_{#1}}
\newcommand\ayOf[1]{\ddot{y}_{#1}}
\newcommand\ayMaxAcc{\ddot{y}_{\mathit{maxAcc}}}
\newcommand\ayMinBr{\ddot{y}_{\mathit{minBr}}}
\newcommand\ayMaxBr{\ddot{y}_{\mathit{maxBr}}}
\newcommand\Safe{S}
\newcommand\fmlPreBr{\fml_{\mathit{preBr}}}
\newcommand\fmlBrake{\fml_{\mathit{Br}}}
\newcommand\AEgoMaxAcc{A_{\egoIndex}^{\mathit{maxAcc}}}
\newcommand\AEgoMinBr{A_{\egoIndex}^{\mathit{minBr}}}
\newcommand\APrecMaxBr{A_{\precIndex}^{\mathit{maxBr}}}
\newcommand\reactionDelay{\rho}
\newcommand\reactionDelayInStep{\reactionDelay_{\mathrm{step}}}
\usepackage{colortbl}
\newcommand{\tbcolor}{\cellcolor{green!25}\bf}
\newcommand{\TIMEOUT}{{\color{red}\textbf{OOM}}}

\ifdefined\VersionWithComments%
 	\definecolor{colorok}{RGB}{80,80,150}
\else
	\definecolor{colorok}{RGB}{0,0,0}
\fi

\newcommand{\eg}{\textcolor{colorok}{e.\,g.,}\xspace}
\newcommand{\ie}{\textcolor{colorok}{i.\,e.,}\xspace}
\newcommand{\st}{\textcolor{colorok}{s.t.}\xspace}

\makeatletter
\def\orcidID#1{\smash{\href{https://orcid.org/#1}{\protect\raisebox{-1.25pt}{\protect\includegraphics{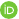}}}}}
\makeatother

\makeatletter
\AtBeginDocument{%
  \@ifpackageloaded{hyperref}
  {\def\@doi#1{\href{https://doi.org/#1}
      {\ttfamily https://doi.org/#1}\egroup}}
  {\def\@doi#1{\ttfamily https://doi.org/#1\egroup}}
  \def\doi{\bgroup\catcode`\_=12\relax\@doi}}
\makeatother

\pdfoutput=1 %

\begin{lncs}
\end{lncs}
\begin{ieee}
\end{ieee}

\title{Oblivious Monitoring for Discrete-Time STL via Fully Homomorphic Encryption}

\author{
\ifdefined\VersionAnonymous%
Anonymous Author(s) %
\else
\LNCSVersion{%
Masaki Waga\inst{1}\orcidID{0000-0001-9360-7490}\and
Kotaro Matsuoka\inst{1}\orcidID{0000-0002-6642-1276}\and
Takashi Suwa\inst{1}\orcidID{0009-0004-9845-7419}\and
Naoki Matsumoto\inst{1}\orcidID{0000-0002-4497-3459}\and
Ryotaro Banno\inst{2}\orcidID{0000-0001-8610-9186}\and
Song Bian\inst{3}\orcidID{0000-0003-0467-6203}\and
Kohei Suenaga\inst{1}\orcidID{0000-0002-7466-8789}%
}%
\IEEEVersion{
 \IEEEauthorblockN{Masaki Waga}
 \IEEEauthorblockA{\textit{Graduate School of Informatics} \\
 \textit{Kyoto University}\\
 Kyoto, Japan \\
 mwaga@fos.kuis.kyoto-u.ac.jp,\orcidID{0000-0001-9360-7490}}
 \and
 \IEEEauthorblockN{Kotaro Matsuoka}
 \IEEEauthorblockA{\textit{Graduate School of Informatics} \\
 \textit{Kyoto University}\\
 Kyoto, Japan \\
 \orcidID{0000-0002-6642-1276}}
 \and
 \IEEEauthorblockN{Takashi Suwa}
 \IEEEauthorblockA{\textit{Graduate School of Informatics} \\
 \textit{Kyoto University}\\
 Kyoto, Japan \\
 email address or ORCID}
 \and
 \IEEEauthorblockN{Naoki Matsumoto}
 \IEEEauthorblockA{\textit{Graduate School of Informatics} \\
 \textit{Kyoto University}\\
 Kyoto, Japan \\
 \orcidID{0000-0002-4497-3459}}
 \and
 \IEEEauthorblockN{Ryotaro Banno}
 \IEEEauthorblockA{\textit{Graduate School of Informatics} \\
 \textit{Cybozu, Inc.}\\
 Osaka, Japan \\
 \orcidID{0000-0001-8610-9186}}
 \and
 \IEEEauthorblockN{Song Bian}
 \IEEEauthorblockA{\textit{Beihang University} \\
 \textit{Beihang University}\\
 Beijing, China \\
 \orcidID{0000-0003-0467-6203}}
 \and
 \IEEEauthorblockN{Kohei Suenaga}
 \IEEEauthorblockA{\textit{Graduate School of Informatics} \\
 \textit{Kyoto University}\\
 Kyoto, Japan \\
 \orcidID{0000-0002-7466-8789}}
}%
\fi
}

\begin{lncs}
 \institute{%
 \ifdefined\VersionAnonymous%
 \else
 Graduate School of Informatics, Kyoto University, Kyoto, Japan%
 \and
 Cybozu, Inc.%
 \and
 Beihang University, Beijing, China%
 \fi
 }
\end{lncs}

\usepackage{graphicx}
\begin{document}

\maketitle
\begin{lncs} 
 \pagestyle{plain}
\end{lncs}

\begin{abstract}
 When monitoring a cyber-physical system (CPS) from a remote server, keeping the monitored data secret is crucial, particularly when they contain sensitive information, \eg{} biological or location data. Recently, Banno et al. (CAV'22) proposed a protocol for online LTL monitoring that keeps data concealed from the server using \emph{Fully Homomorphic Encryption (FHE)}. We build on this protocol to allow \emph{arithmetic} operations over encrypted values, \eg{} to compute a safety measurement combining distance, velocity, and so forth. Overall, our protocol enables oblivious online monitoring of \emph{discrete-time real-valued signals} against signal temporal logic (STL) formulas. Our protocol combines two FHE schemes, CKKS and TFHE, leveraging their respective strengths. We employ CKKS to evaluate arithmetic predicates in STL formulas while utilizing TFHE to process them using a DFA derived from the STL formula. We conducted case studies on monitoring blood glucose levels and vehicles' behavior against the Responsibility-Sensitive Safety (RSS) rules. Our results suggest the practical relevance of our protocol.
\end{abstract}

\begin{ieee}
\end{ieee}

\LongVersion{\section{Introduction}}\ShortVersion{\section{Background}}\label{section:introduction}
\LongVersion{\subsection{Monitoring of CPS with confidential information}}

Given the safety-critical nature of cyber-physical systems (CPSs), monitoring their behavior is crucial, \eg{} to detect undesired behavior and prevent safety-critical situations beforehand.
For instance, if a monitor detects a hardware issue in a car, the car should come to a safe stop immediately.
Monitoring can also enhance the system's comfort.
For example, it can enhance smooth traffic flow by advising drivers on appropriate velocity, which helps reduce traffic congestion.

Such specifications are often confidential since they may include proprietary or sensitive information.
Although they can be kept unknown to the client by monitoring from a remote server, this exposes the monitored behavior to the server.
Such exposure may cause additional security concerns when the monitored behavior contains confidential information, \eg{} biological or location data.

\begin{example}[remote vehicle monitoring]%
 \label{example:self_driving_car_first}
 For smooth and safe driving, it is important to maintain proper velocity and distance from other vehicles.
 Since the desired velocity and distance are typically derived from the driver's visual observation or sensors on the car,
 a challenge arises when visibility is limited, \eg{} on curving roads\LongVersion{ (\cref{figure:RSS_example})}.
 Remote and centralized monitoring can enhance safety, \eg{}
 a remote server receives driving data from each vehicle, conducts monitoring, and sends a precaution to a car if its preceding car is slow.
 However, this approach poses a potential privacy issue because tracking vehicle positions may reveal drivers' personal information (\eg{} their home addresses); hence the monitored data should be kept unknown to the server.
 A more local monitoring approach may resolve this security issue but 
 can introduce another issue: a potential leak of the concrete definition of the desired condition, which may be proprietary. 

\begin{LongVersionBlock}
  \begin{figure}[tbp]
  \centering
  \includegraphics[width=0.90\linewidth]{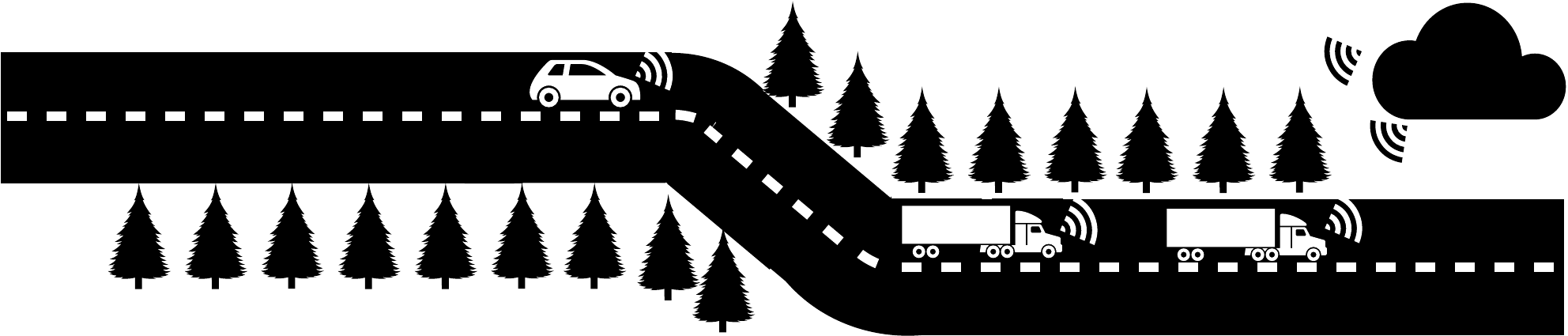}
  \caption{Our leading example: a driver-assistance system with remote monitoring.}%
  \label{figure:RSS_example}
\end{figure}
\end{LongVersionBlock}
\end{example}
\LongVersion{\subsection{Oblivious online LTL monitoring}}
To address this issue, Banno et al.~\cite{DBLP:conf/cav/BannoMMBWS22} introduced a protocol for LTL monitoring on a server without revealing monitored behavior.
Their protocol uses \emph{Fully Homomorphic Encryption (FHE)}~\cite{DBLP:conf/stoc/Gentry09}, which enables computations on ciphertexts.
The client sends a series of ciphertexts representing the monitored behavior to the server, and the server monitors it without decryption.
Also, the monitoring occurs \emph{online}, incrementally processing a stream of ciphertexts.

However, their protocol is limited to LTL formulas over \emph{Boolean} propositions\footnote{More precisely, their protocol supports regular languages over Booleans.},
\ie{} the monitored behavior is limited to a series of Boolean values.
Although their protocol can monitor various temporal behaviors,
it cannot monitor a series of \emph{numbers} against a specification containing \emph{arithmetic} operations.
Particularly, although it can compare numbers with constant thresholds by bit encoding of each value,
it is not straightforward to extend their algorithms to efficiently handle specifications including arithmetic operations over multiple ciphertexts.
Such limitation is due to their choice of the FHE scheme, specifically \emph{FHE over the Torus (TFHE)}~\cite{DBLP:journals/joc/ChillottiGGI20}.
While TFHE excels at Boolean operations, it is known to be slow for arithmetic operations, \eg{} multiplications.

\begin{example}%
 \label{example:self_driving_car_second}
 We continue with \cref{example:self_driving_car_first}.
 Both of the above security concerns are resolved by oblivious monitoring, \ie{}
 remotely monitoring driving data without revealing them to a server.
 However, the protocol in~\cite{DBLP:conf/cav/BannoMMBWS22} %
 cannot compare the vehicles' positions against a desired distance
 computed by arithmetic operations over multiple (encrypted) values, such as their velocities.
\end{example}
\LongVersion{\subsection{Contribution: Oblivious online monitoring of STL formulas with arithmetic operations}\label{section:introduction:oblivious_stl_monitoring}}%
\ShortVersion{\section{Oblivious online discrete-time STL monitoring}\label{section:protocol}}

We enhance the protocol in~\cite{DBLP:conf/cav/BannoMMBWS22} to accommodate specifications incorporating arithmetic operations.
Specifically, our protocol enables oblivious online monitoring of discrete-time real-valued signals against the safety fragment of \emph{signal temporal logic (STL)}~\cite{DBLP:conf/formats/MalerN04} with polynomial constraints as predicates.

\begin{figure}[tbp]
 \centering
 \scalebox{\ShortVersion{0.85}\LongVersion{1.0}}{%
 \begin{tikzpicture}[shorten >=1pt,scale=1.2]
  \node at (3.0,2.5) {\textbf{\underline{Client}}};

  \node[align=center] (sensor) at (1.6,1.75) {\small Monitored\\\small Data};
  \node[align=center,draw=black,rectangle,rounded corners] (encryption) at (4.0,1.75) {\footnotesize Encrypt with\\\footnotesize the CKKS scheme};
  \node[align=center,draw=black,rectangle,rounded corners] (decryption) at (4.0,0.5) {\footnotesize Decrypt with\\\footnotesize the TFHE scheme};
  \node[align=center] (result) at (1.6,0.5) {\small Monitoring\\\small Results};

  \draw[dashed,thick] (5.35,0) -- (5.35,2.7);

  \node at (8.5,2.5) {\textbf{\underline{Server}}};
  \node[align=center,draw=black,rectangle,rounded corners] (ckks_eval) at (7.0,1.75) {\footnotesize Arithmetic Eval. with\\\footnotesize the CKKS scheme};
  \node[align=center,draw=black,rectangle,rounded corners] (as_trlwe) at (10.0,1.2) {\footnotesize Scheme switching \\\footnotesize CKKS $\to$ TFHE};
  \node[align=center,draw=black,rectangle,rounded corners] (homfa) at (7.0,0.5) {\footnotesize Execute a DFA with\\\footnotesize the TFHE scheme~\cite{DBLP:conf/cav/BannoMMBWS22}};

  \path[->] (sensor) edge node{} (encryption)
            (encryption) edge node{} (ckks_eval)
            (ckks_eval) edge node{} (as_trlwe)
            (as_trlwe) edge node{} (homfa)
            (homfa) edge node{} (decryption)
            (decryption) edge node{} (result)
            ;
 \end{tikzpicture}}
 \caption{Our oblivious online STL monitoring protocol\LongVersion{ combining}\ShortVersion{ with} two FHE schemes:\LongVersion{the} CKKS and TFHE\LongVersion{ schemes}.\@ The DFA is constructed from the STL formula\LongVersion{ representing the monitored specification} beforehand.}%
 \label{figure:outline_protocol}
\end{figure}

Our protocol uses the \emph{Cheon-Kim-Kim-Song (CKKS) scheme}~\cite{DBLP:conf/asiacrypt/CheonKKS17} in addition to TFHE to capitalize on their strengths.
Specifically, we use CKKS for arithmetic evaluation and TFHE for logical evaluation.
\cref{figure:outline_protocol} outlines our protocol.
Before starting the protocol, the server constructs a DFA from the STL formula representing the monitored specification.
The DFA receives a series of truth values of the arithmetic predicates in the STL formula.
In our protocol, the client encrypts the monitored data using CKKS and sends the resulting ciphertexts to the server.
The client can use either a public or a private key for encryption.
We remark that multiple clients can participate in our protocol by sharing the public key.
The server uses CKKS to evaluate the polynomial functions in the monitored STL formula.
It then conducts scheme switching to obtain a ciphertext in TFHE representing the truth value of each predicate and executes the DFA using an algorithm from~\cite{DBLP:conf/cav/BannoMMBWS22}.
Finally, the resulting ciphertext is sent to the client and decrypted within TFHE.\@
The server can only use a private key for decryption.

\begin{ShortVersionBlock}
 In summary, the following data are public in our protocol.
\begin{itemize}
 \item The dimension of the monitored signal, \ie{} the set $\Var$ of variables.
 \item For each $\var \in \Var$, the range of $\var$ in the monitored signal\footnote{More precisely, it is sufficient if for each\LongVersion{ predicate} $\fun(\valuation_i) \geq \constant$ in the\LongVersion{ monitored} specification, the server knows the range of $\fun(\valuation_i)$ for any valuation $\valuation_i\colon \Var \to \R$ in the monitored signal.}.
 \item The security parameters of the FHE schemes.
 \item Public keys used for encryption.
 \item Special ciphertexts called \emph{evaluation keys}\LongVersion{ for some FHE operations}\ShortVersion{ used}, \eg{}\ShortVersion{ in} bootstrapping.
\end{itemize}
Notice that an upper bound of the depth of arithmetic operations (in particular, the multiplication) in CKKS can indirectly leak from the security parameters\LongVersion{ of CKKS}.\@

In contrast, we assume that the following data are private.
\begin{itemize}
 \item The values of the monitored signal are concealed from the server.
 \item The private keys are concealed from the server.
 \item The monitored specification is concealed from the client.
\end{itemize}
\end{ShortVersionBlock}

\begin{example}
 We continue with \cref{example:self_driving_car_second}.
 We consider a specification: ``If the (potentially obscured) vehicle ahead at a moderate distance has been driving slowly for more than two seconds, the ego vehicle should start decelerating''.
 The desired velocity can be computed, for example, based on
 the \emph{responsibility-sensitive safety (RSS)} model~\cite{DBLP:journals/corr/abs-1708-06374},
 which models the safe distance as a polynomial of both vehicles' velocities with some vehicle specific parameters as constants.
 Notice that such parameters can be kept secret.
 In our protocol, the vehicles encrypt their positions and velocities using a public key and send them to the server.
 Using CKKS, the server computes the desired distance $\mathit{dist}_{\mathrm{good}}$ and compares it with the observed distance $x_{\mathrm{front}} - x_{\mathrm{ego}}$. Then, the truth value of $x_{\mathrm{front}} - x_{\mathrm{ego}} \geq \mathit{dist}_{\mathrm{good}}$ is switched to TFHE and fed to the DFA constructed from the specification to detect the above undesired situation.
 Finally, an alert system (\eg{} a digital road sign), which we assume has a private key, receives the monitoring result and sends a precaution to the driver.
\end{example}

Our use of the FHE schemes is based on the following observation about FHE-based stream processing:
\begin{quotation}
 \noindent
 The CKKS scheme is suitable for element-wise pre-computation by polynomial arithmetic operations;
 The TFHE scheme is suitable for accumulation over a stream by (potentially) non-polynomial operations, including Boolean operations.
\end{quotation}

The reasoning behind this observation is summarized as follows.\LongVersion{ We will go into more detail in \cref{section:protocol}.}

\begin{itemize}
 \item CKKS is dedicated to but\LongVersion{ highly} efficient for polynomial arithmetic operations, whereas TFHE supports any operations\LongVersion{ encoded by look-up tables (LUTs)}, including Boolean operations.
 \item CKKS requires a bound of the depth of the operations to perform efficiently, while TFHE does not require it.
\end{itemize}

\begin{LongVersionBlock}
 There are several methods to switch from CKKS to TFHE schemes (or similar schemes).
 The existing methods for scheme switching are classified as follows:
 \begin{ienumeration}
 \item highly efficient but introducing errors, \eg{}~\cite{DBLP:journals/jmc/BouraGGJ20,DBLP:conf/sp/LuHHMQ21}, and
 \item highly accurate but slower, \eg{}~\cite{DBLP:journals/pvldb/RenSG0000LZ22,DBLP:conf/ccs/00010PMZJG23}.
 \end{ienumeration}
 The former approach is particularly used to realize non-polynomial but still (Lipschitz) continuous arithmetic operations (\eg{} sigmoid and ReLU in~\cite{DBLP:conf/sp/LuHHMQ21}), where a small error during scheme switching does not significantly change the outcome.
 The latter approach is used to extract a Boolean value via scheme switching, like our usage, since Boolean values flipped by scheme switching can drastically alter the outcome.
 We propose a scheme-switching algorithm in an approach similar to~\cite{DBLP:conf/ccs/00010PMZJG23} with the following improvements:
 \MW{To Kotaro: I updated the following. Can you confirm its correctness?}
 \begin{itemize}
 \item Our algorithm is based on a recent FHE operation called \emph{homomorphic decomposition}~\cite{DBLP:journals/tches/MaHWZW24}, which is known to be more efficient than \emph{homomorphic flooring}~\cite{DBLP:conf/asiacrypt/LiuMP22} used in~\cite{DBLP:conf/ccs/00010PMZJG23} with minor improvements.
 \item We optimize scheme switching by ignoring the ``lower bits'' of the ciphertexts.
 \item We scale the ciphertexts assuming that the server knows the possible range of monitored signals as domain knowledge.
 \end{itemize}
\end{LongVersionBlock}

\begin{ShortVersionBlock} 
\subsection{Scheme switching optimized with value range information}\label{section:protocol:ckks_to_tfhe}
 To combine CKKS and TFHE, we propose a scheme-switching method based on~\cite{DBLP:conf/ccs/00010PMZJG23} with the following improvements:
 \begin{itemize}
  \item Our method is based on a recent FHE operation called \emph{homomorphic decomposition}~\cite{DBLP:journals/tches/MaHWZW24}, which is known to be more efficient than \emph{homomorphic flooring}~\cite{DBLP:conf/asiacrypt/LiuMP22} used in~\cite{DBLP:conf/ccs/00010PMZJG23} with minor improvements.
  \item We optimize scheme switching by ignoring the ``lower bits'' of the ciphertexts.
  \item We scale the ciphertexts assuming that the server knows the possible range of monitored signals as domain knowledge.
\end{itemize}
\end{ShortVersionBlock}

 Conceptually, the scaling normalizes the encrypted values so that errors introduced by scheme switching do not cause overflow (as signed $N$-bit integers).
 Thanks to such normalization, our scheme switching can handle very small values even if we ignore the ``lower bits'' for optimization.
 We decide the scaling factor from the range of signals. %
 \ShortVersion{See Section~3.2 of~\cite{DBLP:journals/corr/abs-2405-16767} %
 for details\LongVersion{ on our scheme switching}.}

\begin{NoSecurityVersion} 
 Based on the client's algorithm above, we propose an oblivious online discrete-time STL monitoring protocol.
 Our protocol operates under specific assumptions on the client and server, which are the same as~\cite{DBLP:conf/cav/BannoMMBWS22}.
 We regard the server as \emph{honest-but-curious}~\cite{DBLP:books/cu/Goldreich2004}, \ie{} it follows the protocol but may obtain the client's private data from the obtained information. In contrast, we regard the client as \emph{malicious}, \ie{} it may deviate from the protocol.
 We prove that our protocol ensures privacy for both parties, \ie{} it conceals the monitored signals and results from the server, and the monitored specification from the client.
\end{NoSecurityVersion}

\begin{LongVersionBlock} 
 We implemented the oblivious discrete-time STL monitoring algorithms in C++20 and evaluated their efficiency through two case studies: monitoring
 \begin{ienumeration}
 \item blood glucose levels and
 \item vehicle behavior against RSS rules.
 \end{ienumeration}
 The experimental results suggest the practical relevance of our protocol.

 Overall, our contributions are summarized as follows.

 \begin{itemize}
 \item We propose an online oblivious monitoring protocol for discrete-time STL.\@
 \item We optimized scheme switching tailored to our purpose.
 \item We show the practical relevance of our protocol via experimental evaluations.
 \end{itemize}
\end{LongVersionBlock}

\begin{LongVersionBlock}

 \paragraph{Related work}\label{sec:related_work}
 \begin{table}[tbp]
 \centering
 \scriptsize
 \caption{Comparison of FHE-based monitoring methods.}%
 \label{table:related_methods}
 \begin{tabular}{c c c c}
  \toprule
  & Used FHE Scheme & Arithmetic operations & Specification is secret\\\midrule
  \textbf{Ours} & CKKS~\cite{DBLP:conf/asiacrypt/CheonKKS17} and TFHE~\cite{DBLP:journals/joc/ChillottiGGI20} & \goodFeature{Yes} & \goodFeature{Yes} \\
  \cite{DBLP:conf/cav/BannoMMBWS22}  & TFHE~\cite{DBLP:journals/joc/ChillottiGGI20} & \badFeature{No} & \goodFeature{Yes} \\ %
  \cite{DBLP:conf/trustbus/TriakosiaRTTSF22} & CKKS~\cite{DBLP:conf/asiacrypt/CheonKKS17} & \goodFeature{Yes} & \badFeature{No}\\ %
  \bottomrule
 \end{tabular}
 \end{table}

 \cref{table:related_methods} summarizes FHE-based monitoring methods.
 As previously mentioned, our method is based on~\cite{DBLP:conf/cav/BannoMMBWS22} and handles arithmetic operations by bridging the CKKS and TFHE schemes.
 Triakosia et al.~\cite{DBLP:conf/trustbus/TriakosiaRTTSF22} proposed a method for oblivious monitoring of manufacturing quality measures with CKKS.\@
 The approach in~\cite{DBLP:conf/trustbus/TriakosiaRTTSF22} is collaborative:
 the server conducts polynomial operations using CKKS, while
 the client conducts non-polynomial operations (\eg{} branching) without using FHE techniques.
 This is done by
 \begin{ienumeration}
 \item decrypting the ciphertexts sent from the server,
 \item conducting non-polynomial operations over plaintexts,
 \item encrypting the result, and
 \item sending it back to the server.
 \end{ienumeration}
 This collaborative approach inherently allows the client to access the monitored specification.
 In contrast, our monitoring algorithm runs entirely on the server, thereby ensuring the specification remains confidential and not exposed to the client.
\end{LongVersionBlock}

\begin{LongVersionBlock} 
 \paragraph{Organization of the paper}\label{section:organization}

 Following an overview of the preliminaries in \cref{sec:preliminaries}, we present our online oblivious STL monitoring protocol in \cref{section:protocol}.
 We show its experimental evaluation in \cref{section:experiments} and conclude in \cref{section:conclusion}.
\end{LongVersionBlock}

\begin{LongVersionBlock}
\section{Preliminaries}\label{sec:preliminaries}

We denote the reals, integers, naturals and positive naturals by $\R$, $\Z$, $\N$, and $\Zp$, respectively.
For $\real \in \R$, we let $\lfloor \real \rfloor \in \Z$ be the maximum integer satisfying $\lfloor \real \rfloor \leq \real$.
For a set $X$, we denote its powerset by $\powerset{X}$,
the set of infinite sequences of $X$ by $X^{\omega}$, and
the set of sequences of $X$ of length $n$ by $X^{n}$, where $n \in \N$.
We let $X^{*} = \bigcup_{n \in \N} X^{n}$ and $X^{\infty} = X^{*} \cup X^{\omega}$.
For $\word \in X^{*}$ and $\word' \in X^{\infty}$, we let $\word \cdot \word' \in X^{\infty}$ be their juxtaposition.
For $\word \in X^{*}$, we let $\prefix{\word}{n}$ be the prefix of $\word$ of length $n$ for $n \leq |\word|$,
where $|\word|$ is the length of $\word$.
We let $\emptyword$ be the empty sequence.
For a DFA $\A$, we denote its language by $\Lg(\A)$.

\subsection{Discrete-time signal temporal logic}\label{subsec:STL}
Let $\Var$ be the finite set of variables.
A (discrete-time) \emph{signal} $\signal \in {(\R^{\Var})}^{\infty}$ is a finite or infinite sequence of functions $\valuation_i\colon\Var \to \R$.
We\LongVersion{ also} call such $\valuation_i$ a (signal) valuation.

\emph{Signal temporal logic (STL)}~\cite{DBLP:conf/formats/MalerN04} is a widely used formalism to represent behaviors of signals.
We use its variant for discrete-time signals.

\begin{definition}
 [signal temporal logic]%
 \label{def:STL}
 For a finite set $\Var$ of variables, the syntax of \emph{signal temporal logic (STL)} is defined as follows, %
where\footnote{In our implementation and experiments, we extend $\fun$ to receive a bounded history of signal valuations, \ie{} $\fun\colon {(\R^{\Var})}^N \to \R$, where $N \in \Zp$ is the history bound.} $\fun\colon \R^{\Var} \to \R$, $\constant \in \R$, and $i,j \in \N \cup \{+\infty\}$ satisfying $i < j$.
\[
 \fml, \fml' \Coloneq \top \mid \fun \geq \constant \mid \neg \fml \mid \fml \lor \fml' \mid \NextOp \fml \mid \fml \UntilOp{[i,j)} \fml'
\]
For an STL formula $\fml$, we let $\Pred(\fml)$ be the set of inequalities $\fun \geq \constant$ in $\fml$.
We define the following as syntax sugar:
$\bot \equiv \neg \top$,
$\fun < \constant \equiv \neg (\fun \geq \constant)$,
$\fml \imply \fml' \equiv (\neg \fml) \lor \fml'$,
$\fml \land \fml' \equiv \neg (\neg \fml \lor \neg \fml')$,
$\DiaOp{[i,j)} \fml \equiv \top \UntilOp{[i,j)} \fml$,
$\BoxOp{[i,j)} \fml \equiv \neg \DiaOp{[i,j)} \neg \fml$,
$\fml \Release{[i,j)} \fml' \equiv \neg (\neg \fml \UntilOp{[i,j)} \neg \fml')$, and
$\fml \TRelease{[i,j)} \fml' \equiv \fml \Release{[i, j)} (\fml \lor \fml')$.
\end{definition}
For an STL formula $\fml$, an \emph{infinite} signal $\signalWithInfiniteInside$, and $k \in \N$,
the satisfaction relation $\satisfy{\signal}{k}{\fml}$ is inductively defined as follows.
\begin{gather*}
 \satisfy{\signal}{k}{\top} \qquad
 \satisfy{\signal}{k}{\fun \geq \constant} \iff \fun(\signal_k) \geq \constant \qquad
 \satisfy{\signal}{k}{\neg \fml} \iff (\signal, k) \not\models \fml \\
 \satisfy{\signal}{k}{\fml \lor \fml'} \iff \satisfy{\signal}{k}{\fml} \lor \satisfy{\signal}{k}{\fml'} \quad
 \satisfy{\signal}{k}{\NextOp{\fml}} \iff \satisfy{\signal}{k+1}{\fml}\\
 \begin{aligned}
  \satisfy{\signal}{k}{\fml \UntilOp{[i,j)} \fml'} &\iff
  \exists l \in \{k+i,k+i+1,\dots,k+j-1\}.\,\satisfy{\signal}{l}{\fml'} \\
  &\qquad\quad\land \forall m \in \{k,k+1,\dots,l-1\}.\, \satisfy{\signal}{m}{\fml}
 \end{aligned}
\end{gather*}

We let $\signal \models \fml$ if we have $\satisfy{\signal}{0}{\fml}$.
An STL formula $\fml$ is\LongVersion{ called} a \emph{safety} property if
for any infinite signal $\signal$
satisfying $\signal \not\models \fml$,
there is a finite prefix $\signal_f$ of $\signal$ such that
for any infinite signal $\signal'$, we have
$\signal_f \cdot \signal' \not\models \fml$.
Such a prefix is called \emph{bad}~\cite{DBLP:journals/fmsd/KupfermanV01}.
Since a discrete-time STL formula $\fml$ is easily representable by an LTL~\cite{DBLP:conf/focs/Pnueli77} formula with $\Pred(\fml)$ as atomic propositions,
violation of a safety STL formula $\fml$ can be monitored by a DFA $\Monitor_{\fml}$ over $\powerset{\Pred(\fml)}$.
Namely, the DFA $\Monitor_{\fml}$ takes a sequence $a_1, a_2, \dots a_n \in {(\powerset{\Pred(\fml)})}^*$
and decides if the corresponding prefix $\signal_f$ of the signal under monitoring is bad or not for $\fml$,
where $a_i = \{\fun \geq \constant \in \Pred(\fml) \mid \fun(\signal_i) \geq \constant\}$.
See, \eg{}~\cite{DBLP:journals/tosem/BauerLS11}, for a construction of such\LongVersion{ a DFA} $\Monitor_{\fml}$ (for LTL).

\subsection{Fully Homomorphic Encryption}\label{subsec:FHE}
\emph{Homomorphic encryption (HE)}~\cite{rivest1978data} is a kind of encryption that enables data evaluation without decryption, \ie{}
for a plaintext $\plain$ and a function $f$,
$f(\plain)$ is (nearly) equal to
$\dec(f_{\mathrm{HE}}(\enc(\plain)))$, where $\enc(x)$ and $\dec(y)$ are the encryption and decryption results
and $f_{\mathrm{HE}}$ is the HE counterpart of $f$.
\emph{Fully HE (FHE)}~\cite{DBLP:conf/stoc/Gentry09} is a kind of HE such that any function $f$ can be evaluated without decryption.

In FHEs, ciphertexts based on the \emph{learning with error (LWE)} problem~\cite{DBLP:journals/jacm/Regev09} and its ring variant \emph{(RLWE)}~\cite{DBLP:journals/jacm/LyubashevskyPR13} are widely used.
Let $n, \quotient \in \Zp$ be security parameters.
An LWE ciphertext $\cipher_{\quotient}$ represents a value in the quotient ring $\qRing$ of $\Z$ modulo $\quotient$.
An RLWE ciphertext $\RLWE{\cipher}_{\quotient, N}$ represents a polynomial of degree $N - 1$ with coefficients in $\qRing$.
We omit $\quotient$ and $N$ if they are clear from the context.
For both LWE and RLWE, a public key $\publicKey$ and a private key $\key$ can be used for encryption while decryption requires the private key.

For security purpose,
small \emph{noise} is added to LWE and RLWE ciphertexts,
and encryption and decryption slightly change the value, \ie{}
for a plaintext $\plain$, $\dec(\enc(\plain))$ is slightly different from $\plain$.
Each FHE scheme has its approach to handle the noise, which we will review later.
However, the noise increases over FHE operations, and eventually, the result of the decryption largely deviates from the expected value, which is considered as a failure.
\LongVersion{For example, for LWE ciphertexts, homomorphic addition
doubles the noise on average.}
The following two approaches are widely used to address this issue:
\begin{ienumeration}
 \item Assuming the number of applications of FHE operations, the noise is sampled so that the decryption is successful throughout the computation;
 \item The noise is reduced by a special FHE operation called \emph{bootstrapping}~\cite{DBLP:conf/stoc/Gentry09} before it becomes too large.
\end{ienumeration}
\begin{table}[tbp]
 \centering
 \scriptsize
 \caption{Comparison of the TFHE and CKKS schemes.}\label{table:TFHE_vs_CKKS}
 \begin{tabular}{lccc}
  \toprule
  & Supported Operations & Our usage & Noise reduction by bootstrapping \\
  \midrule
  CKKS & Polynomial operations & Arithmetic operations & Slow \\
  TFHE & Any (by LUTs) & Logical operations & Relatively fast \\
\bottomrule
\end{tabular}
\end{table}
\begin{table*}[tbp]
 \centering
 \scriptsize
 \caption{Summary of the ciphertexts used in our protocol.}\label{table:ciphertexts}
 \begin{tabular}{lcccccc}
  \toprule
  & Encoded values & Encoded values & Example usage in & Notation in \\
  & in CKKS & in TFHE & our protocol & this paper\\
  \midrule
  LWE & --- & Boolean value & Monitoring results & bold font, \eg{} $\cipher$ \\
  RLWE & (Approx.) Reals & Boolean array & Signal valuations & bold font + overline, \eg{} $\RLWE{\cipher}$ \\
  RGSW & --- & Boolean value & Truth values of predicates & bold font + tilde, \eg{} $\RGSW{\cipher}$ \\
  \bottomrule
 \end{tabular}
\end{table*}

Among various FHE schemes, \emph{Cheon-Kim-Kim-Song scheme (CKKS)}~\cite{DBLP:conf/asiacrypt/CheonKKS17} and \emph{FHE over the Torus (TFHE)}~\cite{DBLP:journals/joc/ChillottiGGI20} are two of the most widely used schemes.
\cref{table:TFHE_vs_CKKS} summarizes their comparison.
\cref{table:ciphertexts} summarizes the encoded value for each kind of ciphertexts in each scheme and our specific usage.

In CKKS, RLWE ciphertexts are usually used.
CKKS is typically used to approximately encode real values and apply polynomial operations, \eg{} addition, subtraction, and multiplication.
When representing a real value $\real \in \R$ by an RLWE ciphertext\footnote{For simplicity, we omit the embedding of multiple values as roots of a polynomial~\cite{DBLP:conf/asiacrypt/CheonKKS17} for vectorization, which we do not use in our implementation.},
we first approximate $\real$ by a multiplication of $\plain \in \qRing$ and a positive floating-point number $\scale$, where $\plain$ represents a bit-encoding of a signed integer.
Then, we use a pair $(\scale, \RLWE{\cipher}_{q, N})$ of $\scale$ and an RLWE ciphertext $\RLWE{\cipher}_{q, N}$ encrypting a polynomial $\RLWE{\plain}$ with constant term $\plain$ to represent $\real$.
Notice that $\scale$ is not encrypted because it is typically chosen independently of $\real$.
We usually omit $\scale$ and simply write $\RLWE{\cipher}_{q, N}$.
We let $\encCKKS$ and $\decCKKS$ be the encryption and decryption with the above encoding.
In CKKS, errors are simply ignored, assuming that they are much smaller than, \eg{} noise in sampling and numeric error.
While CKKS is efficient for polynomial operations, it does not support non-polynomial operations.
Bootstrapping in CKKS is known to be slow~\cite{DBLP:conf/eurocrypt/CheonHKKS18}, and thus, CKKS is typically employed in situations where an upper bound of the number of applications of FHE operations is known beforehand.
In TFHE, both LWE and RLWE ciphertexts are used.
TFHE supports a multiplexer operation: an FHE operation to select one of the values of the given ciphertexts according to the value of another ciphertext (called a control bit).
By combining multiplexers, one can implement a look-up table (LUT), which allows to encode any operations with TFHE.\@
In this paper, we use TFHE to encode Boolean values and logical operations, such as AND and OR.\@
Specifically, we use LWEs to represent Boolean values and RLWEs to represent Boolean arrays.
One typical encoding used for LWEs is such that $\plain \in \qRing$ represents $\top$ if and only if $\plain \in \{0, 1, \dots, \quotient / 2 - 1\}$.
The encoding for RLWEs is similar: we use coefficients as an array of values.
The result of the decryption (as a Boolean value) is successful if
the noise is small\LongVersion{ enough,} and the plaintext is in the expected range.
We let $\encTFHE$ and $\decTFHE$ be the encryption and decryption with the above encoding.
Although TFHE supports any operations, it is not as fast as CKKS for polynomial operations, \eg{} a single multiplication of two 16-bit integer values takes more than a few seconds~\cite{DBLP:journals/joc/ChillottiGGI20}.
Bootstrapping is relatively fast in TFHE, and\LongVersion{ thus,} TFHE is suitable for handling unbounded length of data, reducing the noise by bootstrapping.
Another type of ciphertexts called \emph{RGSW}~\cite{DBLP:conf/crypto/GentrySW13} are also used in TFHE.\@
An RGSW ciphertext is, roughly speaking, a collection of RLWE ciphertexts.
When conducting logical operations\LongVersion{ with TFHE}, an RGSW ciphertext represents a Boolean value that is used, \eg{} as the control bit of multiplexers.

\subsection{Online algorithm for oblivious DFA execution}\label{subsec:HomFA}
Banno et al.~\cite{DBLP:conf/cav/BannoMMBWS22} proposed two TFHE-based algorithms (\REVERSE{} and \BLOCK{}) to obliviously execute a DFA over $\{\top, \bot\}$.
They use these algorithms for \emph{online} LTL monitoring by first constructing a DFA $\A$ over $\{\top, \bot\}$ from an LTL formula $\fml$ over atomic propositions $\AP$.
The construction is by
\begin{ienumeration}
 \item making a DFA over $\powerset{\AP}$ from $\fml$, \eg{} with~\cite{DBLP:journals/tosem/BauerLS11}, and
 \item modifying its alphabet to $\{\top, \bot\}$ by encoding each $\action^{\AP} \in \powerset{\AP}$ by $\action_1, \action_2,\dots,\action_{|\AP|} \in {\{\top, \bot\}}^{|\AP|}$.
\end{ienumeration}

Given a sequence $\encWordSequence$ of RGSW ciphertexts representing the input word $\wordWithInside \in {\{\top, \bot\}}^*$,
their algorithms return a sequence of LWE ciphertexts representing if prefixes of the input word are accepted by $\A$.
Their algorithms incrementally process the input word and return prefixes of the result before obtaining the entire input.
Thus, they are suitable for online monitoring.

There algorithms are based on another TFHE-based algorithm in~\cite{DBLP:conf/cav/BannoMMBWS22} for \emph{offline} execution of DFAs.
The offline algorithm reads the input word $\word$ \emph{from the tail} to achieve DFA execution only using multiplexers, which is fast for TFHE.\@
In \REVERSE{}, the offline algorithm is applied to the reversed DFA $\reverse{\A}$, \ie{} the minimum DFA such that $\word \in \Lg(\A) \iff \reverse{\word} \in \Lg(\reverse{\A})$, where $\reverse{\word}$ is the reversed word of $\word$.
\REVERSE{} returns the result for each prefix $\prefix{\word}{i}$ of $\word$.

In \BLOCK{}, they use a variant of the offline algorithm that also returns an RLWE ciphertext representing the state after reading the given ciphertexts.
The input word $\word$ is split into a series of words $\word_1, \word_2, \dots, \word_N \in {({\{\top, \bot\}}^{\BlockSize})}^{*}$ of length $\BlockSize \in \Zp$,
each word $\word_j$ is fed to the modified offline algorithm from $j = 1$ to $j = N$,
and the result for each block is used to choose the state after reading $\word_1, \word_2, \dots, \word_j$.
Since \BLOCK{} feeds each $\word_j \in {\{\top, \bot\}}^{\BlockSize}$ to the offline algorithm,
it returns the result for prefixes $\prefix{\word}{i}$ such that $i = \BlockSize \times j$ for some $j \in \N$.

While \REVERSE{} has better scalability with respect to the size of the original LTL formula in terms of the worst case complexity,
\cite{DBLP:conf/cav/BannoMMBWS22} reports that their practical superiority is inconclusive:
\REVERSE{} is typically faster than \BLOCK{} but \REVERSE{} takes much longer time, \eg{} when the reversed DFA $\reverse{\A}$ is huge.

\end{LongVersionBlock}

\begin{LongVersionBlock}
 \section{Oblivious online discrete-time STL monitoring}\label{section:protocol}
 We propose a protocol for oblivious online STL monitoring based on the algorithms in\ShortVersion{~\cite{DBLP:conf/cav/BannoMMBWS22}}\LongVersion{ \cref{subsec:HomFA}}.
 In our protocol, the following data are public.
 \begin{itemize}
 \item The dimension of the monitored signal, \ie{} the set $\Var$ of variables.
 \item For each $\var \in \Var$, the range of $\var$ in the monitored signal\footnote{More precisely, it is sufficient if for each\LongVersion{ predicate} $\fun(\valuation_i) \geq \constant$ in the\LongVersion{ monitored} specification, the server knows the range of $\fun(\valuation_i)$ for any valuation $\valuation_i\colon \Var \to \R$ in the monitored signal.}.
 \item The security parameters of the FHE schemes.
 \item The scaling factor $\scale$ of each RLWE ciphertext within CKKS.
 \item Public keys used for encryption.
 \item Special ciphertexts called \emph{evaluation keys}\LongVersion{ for some FHE operations}\ShortVersion{ used}, \eg{}\ShortVersion{ in} bootstrapping.
 \end{itemize}
Notice that an upper bound of the depth of the arithmetic operations (in particular, the multiplication) in CKKS can indirectly leak from the security parameters of CKKS.\@

 In contrast, we assume that the following data are private.
 \begin{itemize}
 \item The values of the monitored signal are concealed from the server.
 \item The private keys are concealed from the server.
 \item The monitored specification is concealed from the client.
 \end{itemize}

 From the server's perspective, the primary task in this protocol is formulated as follows, where the public data are omitted.
 We mainly focus on this\LongVersion{ server-side's} problem.

 \defProblem{Oblivious safety discrete-time STL monitoring}{%
 A sequence $\encSignalSequence$ of RLWE ciphertexts encrypting the monitored signal $\signal \in {(\R^{\Var})}^{*}$ of length $\signalLength$, a discrete-time STL formula $\fml$ representing a safety property, and the refresh interval $\BlockSize \in \Zp$ of the results}{%
 A sequence $\encResultSequence$ of LWE ciphertexts \st{} for each $i \in \{1,2,\dots,n\}$, $\decTFHE(\encResult_{i}) = \top$ if and only if $\prefix{\signal}{\lfloor i/\BlockSize \rfloor \BlockSize}$ is a bad prefix for $\fml$}

 \subsection{Overview of our oblivious online STL monitoring algorithm}\label{section:protocol:overview}
 \SetKwFunction{FMakeMonitor}{makeDFA}
 \SetKwFunction{FInitializeMonitor}{initializeDFAExec}
 \SetKwFunction{FEvalCKKS}{evalCKKS}
 \SetKwFunction{FCKKSToTFHE}{isPositive}
 \SetKwFunction{FFeed}{feed}
 \SetKwFunction{FIsAccepting}{isAccepting}
 \begin{algorithm}[tbp]
 \caption{Outline of our oblivious online discrete-time STL monitoring algorithm with CKKS and TFHE schemes.}%
 \label{algorithm:protocol:outline}
 \DontPrintSemicolon{}
 \scriptsize
 \newcommand{\myCommentFont}[1]{\texttt{\scriptsize{#1}}}
 \SetCommentSty{myCommentFont}
 \KwIn{A safety STL formula $\fml$, $\BlockSize \in \Zp \cup \{{*}\}$, and RLWE ciphertexts $\encSignalSequence$ encrypting the monitored signal $\signal = \signalFiniteInside \in {(\R^{\Var})}^{*}$ within CKKS}
 \KwOut{LWE ciphertexts $\encResultSequence$ within TFHE such that for each $i \in \{1,2,\dots,n\}$,
 if $\BlockSize = *$, $\decTFHE(\encResult_i) = \top$ if and only if $\prefix{\signal}{i}$ is bad for $\fml$, and otherwise,
 $\decTFHE(\encResult_i) = \top$ if and only if $\prefix{\signal}{\lfloor i/\BlockSize \rfloor \BlockSize}$ is bad for $\fml$
 }
 $\Monitor_{\fml} \gets \FMakeMonitor(\fml)$\tcp*[r]{We have $\Lg(\Monitor_{\fml}) = \{\BitEncoding(\signal) \mid \text{$\signal$ is bad for $\fml$}\}$.}\label{algorithm:protocol:outline:dfa-construction}
 \tcp{Initialize with one of the algorithms in~\cite{DBLP:conf/cav/BannoMMBWS22}}
 \If(\tcp*[f]{We use $\BlockSize = {*}$ to select \REVERSE{}, where refresh interval is always 1.}){$\BlockSize = {*}$} {
   $\mathsf{DFAExec} \gets \FInitializeMonitor(\REVERSE{}, \Monitor_{\fml})$\label{algorithm:protocol:outline:monitor-reverse-initialization}\;
 } \Else{
   $\mathsf{DFAExec} \gets \FInitializeMonitor(\BLOCK{}, \Monitor_{\fml}, \BlockSize)$\label{algorithm:protocol:outline:monitor-block-initialization}\;
 }
 \For{$i \gets 0$ \KwTo{} $\signalLength - 1$} {\label{algorithm:protocol:outline:signal-loop-begin}
   \For{$\fun \geq \constant \in \Pred(\fml)$} {\label{algorithm:protocol:outline:pred-loop-begin}
     $\diffCKKS \gets \FEvalCKKS(\fun(\encSignal_{1 + i \times |\Var|}, \encSignal_{2 + i \times |\Var|}, \dots, \encSignal_{|\Var| + i \times |\Var|}) - \constant)$\label{algorithm:protocol:outline:CKKS-diff}\;
     \tcp{scheme switching; $\decTFHE(\signTFHE) = \top \iff \fun(\signal_i) \geq \constant$ holds.}
     $\signTFHE \gets \FCKKSToTFHE(\diffCKKS)$\label{algorithm:protocol:outline:ckks-to-tfhe}\;
     \FFeed{$\mathsf{DFAExec}$, $\signTFHE$}\label{algorithm:protocol:outline:feed-tlwe}\;
   }
   $\encResult_i \gets \FIsAccepting(\mathsf{DFAExec})$\label{algorithm:protocol:outline:get-result}\;
 }
 \end{algorithm}

 \cref{algorithm:protocol:outline} outlines our algorithm for oblivious online STL monitoring.
 We use RLWE ciphertexts within CKKS for the inputs $\encSignalSequence$ and LWE ciphertexts within TFHE for the outputs $\encResultSequence$.
 First, we convert the given STL formula $\fml$ to a DFA $\Monitor_{\fml}$ satisfying
 $\Lg(\Monitor_{\fml}) = \{\BitEncoding(\signal) \mid \text{$\signal$ is bad for $\fml$}\}$, where $\BitEncoding$ is the function to encode ${(\R^{\Var})}^*$ by ${\{\top, \bot\}}^*$ by evaluating $\Pred(\fml)$ (\cref{algorithm:protocol:outline:dfa-construction}).
 Then, we initialize one of the oblivious online DFA execution algorithms (\ie{} \REVERSE{} or \BLOCK{}) in~\cite{DBLP:conf/cav/BannoMMBWS22} (\cref{algorithm:protocol:outline:monitor-reverse-initialization,algorithm:protocol:outline:monitor-block-initialization}) using $\BlockSize$.

 The main loop of monitoring starts at \cref{algorithm:protocol:outline:signal-loop-begin}.
 For each predicate $\fun \geq \constant$ in $\Pred(\fml)$, we obtain a ciphertext representing the truth value of $\fun \geq \constant$ for $\signal_i$ by
 homomorphically computing $\fun(\signal_i) - \constant$ with CKKS (\cref{algorithm:protocol:outline:CKKS-diff}) and
 constructing an RGSW ciphertext $\signTFHE$ with TFHE satisfying $\decTFHE(\signTFHE) = \top \iff \fun(\signal_i) \geq \constant$ (\cref{algorithm:protocol:outline:ckks-to-tfhe}).
 Then, we feed $\signTFHE$ to the DFA execution algorithm (\cref{algorithm:protocol:outline:feed-tlwe}) and
 use the LWE ciphertext $\FIsAccepting(\mathsf{DFAExec})$ representing if the monitor is at the accepting state as the result $\encResult_i$ (\cref{algorithm:protocol:outline:get-result}).
 Notice that when we use \BLOCK{}, the result $\encResult_i$ is periodically updated
 because $\mathsf{DFAExec}$ consumes the given ciphertexts and updates its internal state only if $i = j \BlockSize$ for some $j \in \Zp$.

 Our usage of CKKS and TFHE is justified as follows.
 As we mentioned in \cref{subsec:FHE}, for polynomial arithmetic operations, CKKS is superior to TFHE, whereas logical operations require TFHE.\@
 Moreover, since the signal length $\signalLength$ is unknown to the server when the monitoring starts, bootstrapping is necessary to execute the DFA $\Monitor_{\fml}$, which is advantageous for TFHE.\@
 In contrast, the depth of the FHE operations to evaluate $\fun(\encSignal_{1 + i \times |\Var|}, \encSignal_{2 + i \times |\Var|}, \dots, \encSignal_{|\Var| + i \times |\Var|}) - \constant$ is fixed from $\fml$, and bootstrapping is unnecessary if the security parameters are appropriately configured.
\end{LongVersionBlock}
\ShortVersion{ \SetKwFunction{FEvalCKKS}{evalCKKS}} 
\SetKwFunction{FHomDecomp}{HomDecomp}

\begin{LongVersionBlock}
\subsection{Scheme switching optimized with value range information}\label{section:protocol:ckks_to_tfhe}
 \SetKwFunction{FSampleExtract}{RLWEToLWE}
 \SetKwFunction{FCircuitBootstrapping}{LWEToRGSW}
 \SetKwFunction{FModSwitch}{modSwitch}
 \SetKwFunction{FPartialHomDecomp}{PartHomDecomp}
 \begin{algorithm}[tbp]
 \caption{Outline of \texttt{isPositive}.}%
 \label{algorithm:protocol:ckks_to_tfhe}
 \DontPrintSemicolon{}
 \scriptsize
 \newcommand{\myCommentFont}[1]{\texttt{\scriptsize{#1}}}
 \SetCommentSty{myCommentFont}
 \KwIn{An RLWE ciphertext $\diffCKKS$, its scaling factor $\scale \in \R$ for CKKS, the maximum value $\diffMax \in \Rp$ of $|\decCKKS(\diffCKKS)|$, the switching size $k \in \Zp$ in $\FHomDecomp$, and the modulus $\quotientCKKS$ and $\quotientTFHE$ of\LongVersion{ the quotient ring in} the input and output ciphertexts}
 \KwOut{An RGSW ciphertext $\signTFHE[\RGSW]$ such that if $\signTFHE[\RGSW]$ and $\diffCKKS$ are decrypted correctly, $\decTFHE(\signTFHE[\RGSW]) = \top \iff \decCKKS(\diffCKKS) \geq 0$ holds}
 $\diffScaledCKKS \gets \FEvalCKKS(\diffCKKS \times \frac{\lfloor \quotientCKKS / 2 - 1\rfloor}{\diffMax \times \scale})$\tcp*[r]{Scale the value space keeping the sign}\label{algorithm:protocol:ckks_to_tfhe:amplify}
 $\diffSwitched[\RLWE] \gets \FModSwitch(\diffScaledCKKS, \quotientCKKS \to \quotient_{\mathrm{tmp}})$\tcp*[r]{Switch to the intermediate modulus $\quotient_{\mathrm{tmp}}$}\label{algorithm:protocol:ckks_to_tfhe:rescaling}
 $\diffSwitched \gets \FSampleExtract(\diffSwitched[\RLWE])$\label{algorithm:protocol:ckks_to_tfhe:sample_extract}\;
 $\signTFHE[] \gets \FHomDecomp(\diffSwitched, k)$\tcp*[r]{Switch the security parameters\LongVersion{ for TFHE} preserving the MSB}\label{algorithm:protocol:ckks_to_tfhe:key_switching}
 $\signTFHE \gets \FCircuitBootstrapping(\signTFHE[])$\label{algorithm:protocol:ckks_to_tfhe:circuit_bootstrapping}\;
 \end{algorithm}

 \MW{To Kotaro: Please read this section}
 \LongVersion{Among the steps in \cref{algorithm:protocol:outline},
 the most non-trivial and computationally demanding one is \texttt{isPositive} at \cref{algorithm:protocol:outline:ckks-to-tfhe},
 the construction of an RGSW ciphertext in TFHE.\@
 We elaborate it here.}

 As we explained in \cref{subsec:FHE},
 when we encrypt $\real \in \R$ within CKKS,
 we use $\plain \in \qRing$ satisfying %
 $\real \geq 0$ if and only if $\plain \in \{0, 1, \dots, \lfloor \quotient / 2 - 1\rfloor\}$.
 Since this coincides with the encoding of $\top$ within TFHE if security parameters are appropriately switched,
 we can obtain an RGSW ciphertext $\signTFHE$ satisfying $\decTFHE(\signTFHE) = \top \iff \fun(\signal_i) \geq \constant$ by
 converting the given RLWE ciphertext to an RGSW ciphertext with appropriate switching of security parameters.
 In this construction, switching of security parameters is the most computationally demanding
 even if we use a recent FHE operation called \emph{homomorphic decomposition} ($\FHomDecomp$)~\cite{DBLP:journals/tches/MaHWZW24}, which is more efficient than the operation used in~\cite{DBLP:conf/ccs/00010PMZJG23}.
 Thus, we optimize this step to improve the efficiency of the overall algorithm.
 Intuitively, $\FHomDecomp$ slices the binary representation of an integer into multiple short bit sequences.
 Given an LWE ciphertext $\cipher_{\quotient}$ with modulus $\quotient$ and $W, k \in \Zp$,
 $\FHomDecomp$ decomposes the first $W \times k$ bits of $\cipher_{\quotient}$ into LWE ciphertexts $\cipher_{\quotient', 1}, \cipher_{\quotient', 2}, \dots, \cipher_{\quotient', k}$ with modulus $\quotient'$ such that
 the first $W$ bits of each $\cipher_{\quotient', i}$ represents the ``$i$-th block'' of $\cipher_{\quotient}$, where each block consists of $W$ bits.
 The decomposed ciphertexts may have different security parameters from $\cipher_{\quotient}$.

 Although the result of decryption as a Boolean value solely depends on the ``most significant bit (MSB)'' of $\plain$,
 naively, we need to decompose all the blocks of $\cipher_{\quotient}$.
 This is because if $\plain$ is near the boundaries of flipping (\ie{} $0$, $\lfloor \quotient / 2 - 1 \rfloor$, $\lfloor \quotient / 2 - 1 \rfloor + 1$, and $\quotient - 1$),
 the higher blocks are even closer to the boundaries and quite fragile to the noise by scheme switching.
 Nevertheless, if we have the possible range of the encrypted raw value $\plain$ as domain knowledge,
 we can scale $\plain$ so that it is sufficiently far from the boundaries.
 The range of $\plain$ can be estimated from the range of each $\var \in \Var$ in the monitored signal, $\scale$ in CKKS, and the monitored formula $\fml$.
 With such preprocessing, we apply $\FHomDecomp$ only partly, \eg{} to the higher half of the blocks.

 \cref{algorithm:protocol:ckks_to_tfhe} outlines \texttt{isPositive}.
 In \cref{algorithm:protocol:ckks_to_tfhe:amplify}, we scale $\diffCKKS$ so that the range of $\decCKKS(\diffCKKS)$ almost matches the range\LongVersion{ $[- \lfloor \quotientCKKS / 2\rfloor / \scale, \lfloor \quotientCKKS / 2 - 1\rfloor / \scale]$} of the ciphertext.
 In \cref{algorithm:protocol:ckks_to_tfhe:rescaling}, we switch the modulus from $\quotientCKKS$ to $\quotientTFHE$.
 In \cref{algorithm:protocol:ckks_to_tfhe:sample_extract}, we use an FHE operation called \emph{sample extraction}~\cite{DBLP:journals/joc/ChillottiGGI20}
 to obtain an LWE ciphertext $\diffSwitched$ satisfying $\decTFHE(\diffSwitched) = \top \iff \decCKKS(\diffSwitched[\RLWE]) \geq 0$ if the decryption is successful.
 In \cref{algorithm:protocol:ckks_to_tfhe:key_switching}, we apply homomorphic decomposition to obtain another LWE ciphertext $\signTFHE[]$ with different security parameters, preserving the decryption result if it is successful.
 Lastly, in \cref{algorithm:protocol:ckks_to_tfhe:circuit_bootstrapping}, we use another FHE operation called \emph{circuit bootstrapping}~\cite{DBLP:journals/joc/ChillottiGGI20} to translate $\signTFHE[]$ to an RGSW ciphertext $\signTFHE$, preserving the decryption result if it is successful. %
 Notice that even if we do not have $\diffMax$, we can still use this algorithm by 
 \begin{ienumeration}
 \item skipping \cref{algorithm:protocol:ckks_to_tfhe:amplify} and
 \item applying $\FHomDecomp$ to all the bits, \ie{} $k = \quotient_{\mathrm{tmp}}$ in \cref{algorithm:protocol:ckks_to_tfhe:key_switching}.
 \end{ienumeration}

\end{LongVersionBlock}

\begin{LongVersionBlock} 
 \subsection{Correctness and complexity of our algorithm}\label{section:protocol:properties}

 In what follows, we assume that decryption is always successful.
 From the discussion in \cref{section:protocol:ckks_to_tfhe},
 we have the following correctness of \cref{algorithm:protocol:ckks_to_tfhe}.

 \begin{lemma}
 [correctness of \cref{algorithm:protocol:ckks_to_tfhe}]%
 \label{lemma:ckks_to_tfhe:correctness}
 Given an RLWE ciphertext $\diffCKKS$,
 its scaling factor $\scale$,
 $\diffMax \in \Rp$, 
 the switching size $k \in \Zp$, and
 the modulus $\quotientCKKS$ and $\quotientTFHE$ of\LongVersion{ the quotient ring used in} the input and output ciphertexts,
 if $|\decCKKS(\diffCKKS)| < \diffMax$ holds,
 \cref{algorithm:protocol:ckks_to_tfhe} returns an RGSW ciphertext $\signTFHE$ such that we have
 $\decCKKS(\diffCKKS) \geq 0 \iff \decTFHE(\signTFHE) = \top$.
 \qed{}
 \end{lemma}

 The correctness of \cref{algorithm:protocol:outline} is immediate from \cref{lemma:ckks_to_tfhe:correctness} and the correctness of \REVERSE{} and \BLOCK{} in~\cite{DBLP:conf/cav/BannoMMBWS22}.

 \begin{theorem}
 [correctness]
 Given a safety STL formula $\fml$, $\BlockSize \in \Zp \cup \{*\}$, and RLWE ciphertexts $\encSignalSequence$ encrypting a signal $\signal \in {(\R^{\Var})}^{*}$,
 if $|\fun(\signal_i) - \constant| < \diffMax$ holds for any $\fun \geq \constant \in \Pred(\fml)$ and $i \in \{1, 2,\dots, \signalLength\}$,
 \cref{algorithm:protocol:outline} returns LWE ciphertexts $\encResultSequence$ such that for each $i \in \{1,2,\dots, n\}$, 
 \begin{ienumeration}
  \item if $\BlockSize = *$, $\decTFHE(\encResult_i) = \top$ if and only if $\prefix{\signal}{i}$ is bad for $\fml$ and
  \item otherwise, $\decTFHE(\encResult_{i}) = \top$ if and only if $\prefix{\signal}{\lfloor i / \BlockSize \rfloor \BlockSize}$ is bad for $\fml$.
 \end{ienumeration}
 \end{theorem}
 \begin{proof}[sketch]
  Since we have $\BitEncoding(\signal) \in \Lg(\Monitor_{\fml})$ if and only if $\signal$ is bad for $\fml$, from~\cite[Theorems 2 and 3]{DBLP:conf/cav/BannoMMBWS22}, it suffices to show $\fun(\prefix{\signal}{i}) \geq \constant \iff \decTFHE(\signTFHE) = \top$ in the loop from \cref{algorithm:protocol:outline:pred-loop-begin}.
  This holds by \cref{lemma:ckks_to_tfhe:correctness}.
  \qed{}
 \end{proof}
 \begin{table}[tb]
 \centering
 \caption{Complexity of \cref{algorithm:protocol:outline} for the size $|\fml|$ of the STL formula, the length $n$ of the monitored signal.}%
 \label{table:protocol:complexity}
 \scriptsize
 \begin{tabular}{ccccc}
  \toprule
  \multirow{2}{*}{Algorithm for DFA Execution} & \multicolumn{2}{c}{Number of Applications} & \multirow{2}{*}{Scheme Switching}\\ %
    & TFHE & CKKS & \\\midrule
  \REVERSE{} & $O(n 2^{|\fml|})$     & $O(n |\fml|)$ & $O(n |\fml|)$\\ %
  \BLOCK{}   & $O(n 2^{2^{|\fml|}})$ & $O(n |\fml|)$ & $O(n |\fml|)$\\ %
  \bottomrule
 \end{tabular}
 \end{table}

 \cref{table:protocol:complexity} summarizes the complexity of \cref{algorithm:protocol:outline}.
 The number of applications of TFHE operations is immediately obtained from the discussion in~\cite{DBLP:conf/cav/BannoMMBWS22}.
 For each input ciphertext $\encSignal_i$,
 the number of applications of CKKS operations is the total number of arithmetic operators in $\Pred(\fml)$, and
 the number of scheme switching is $|\Pred(\fml)|$, both are bounded by $|\fml|$.
 \MW{To claim that the space complexity is the same as~\cite{DBLP:conf/cav/BannoMMBWS22}, we need to discuss the space complexity due to CKKS. I do not think we have enough space and time for it.}
 Overall, \cref{algorithm:protocol:outline} has linear time complexity %
 with respect to $n$, suggesting the scalability of our algorithm.
\end{LongVersionBlock}

\begin{NoSecurityVersion}
\subsection{Security of the two-party protocol for monitoring}\label{section:protocol:security_discussion}
 \reviewer{5}{I found the discussion on the security of the overall oblivious monitoring scheme a bit superfluous. It essentially throws it back to the
 IND-CPA-security argument. However, as the authors mention, CKKS introduces error with repeated application, there are a number of novel
 aspects introduced by switching between the schemes, and other algorithmic choices. It is not clear to me how this impacts the overall security of the
 scheme, and the small explanation from the authors seems inadequate.}
 Here, we briefly discuss the security of our two-party protocol for oblivious online STL monitoring.
 \LongVersion{
 \subsubsection{Threat model}
 }
 We employ the same threat model as~\cite{DBLP:conf/cav/BannoMMBWS22}:
 We regard the client as \emph{malicious}\LongVersion{, \ie{} it may deviate from the protocol;
 We regard}\ShortVersion{ and} the server as \emph{honest-but-curious}~\cite{DBLP:books/cu/Goldreich2004}, \ie{} it adheres to the protocol but may obtain the client's private data from the obtained information.
 \LongVersion{The public and private data are summarized in the beginning of \cref{section:protocol}.}

 \LongVersion{\subsubsection{Privacy of the client}}
 Through the protocol, the server obtains RLWE ciphertexts $\encSignalSequence$ and the scaling factor $\scale_i \in \R$ for each $i \in \{1,2,\dots,n\}$.
 Since both CKKS and TFHE schemes are IND-CPA-secure~\cite{DBLP:journals/iacr/CheonHK20,DBLP:journals/joc/ChillottiGGI20} and scheme switching is a special case of key switching,
 our protocol's
 IND-CPA security follows from~\cite[Section 4]{DBLP:conf/pkc/MicciancioV24} under the \emph{quadratic circular LWE assumption}, which is common in FHE schemes.
 Therefore, our protocol keeps the client's private data secret.
 \LongVersion{\subsubsection{Privacy of the server}}

 In our protocol, the client can arbitrarily feed input RLWE ciphertexts $\encSignalSequence$ and obtain the corresponding LWE ciphertexts $\encResultSequence$.
 From the decryption results of these ciphertexts, the client cannot exactly identify the specification $\fml$, in particular, the DFA $\Monitor_{\fml}$ and the predicates $\fun \geq \constant$ in $\Pred(\fml)$, because there are multiple (non-equivalent) DFAs or predicates consistent with the obtained information.
 Notice that they can be identified with additional information, for example, the size of $\Monitor_{\fml}$~\cite{DBLP:journals/iandc/Angluin81}.

 Nevertheless, some information might leak from the noise in each $\encResult_i$ because it can depend on the conducted FHE operations.
 To fully ensure the privacy of the server, the server must randomize the noise in the ciphertexts $\encResult_i$ without changing the decryption result, which is also done in~\cite{DBLP:conf/cav/BannoMMBWS22}.
 From discussion in~\cite{DBLP:conf/cav/BannoMMBWS22}, the server's private data are also kept secret, assuming the \emph{shielded randomness leakage (SRL)} security~\cite{DBLP:conf/icalp/BrakerskiDGM22,DBLP:conf/stoc/GayP21} over TFHE.\@

\end{NoSecurityVersion}

\section{Experimental Evaluation}\label{section:experiments}

We experimentally evaluated the practicality of our oblivious online STL monitoring protocol with our prototype toolkit \ourTool{}\footnote{\ourTool{} is publicly available at \url{https://github.com/MasWag/arith_homfa}.}.
We implemented \ourTool{} in C++20.
We use Microsoft SEAL~\cite{sealcrypto} and TFHEpp~\cite{DBLP:conf/uss/MatsuokaBMS021} as the libraries for CKKS and TFHE, respectively.
We used Spot~2.11.5~\cite{DBLP:conf/cav/Duret-LutzRCRAS22} to handle temporal logic formulas.

We aim to address the following research questions.
\begin{itemize}
 \item[RQ1] How fast is the proposed workflow? Is\LongVersion{ the throughput}\ShortVersion{ it} sufficient for practical usage?
 \item[RQ2] How does the optimization in \cref{section:protocol:ckks_to_tfhe} improve the efficiency?
 \item[RQ3] Does the use of CKKS improve the workflow's efficiency, for the benchmarks tractable by the approach in~\cite{DBLP:conf/cav/BannoMMBWS22}, which only uses\LongVersion{ the} TFHE\LongVersion{ scheme}?
 \item[RQ4] Is the computational demand of the client low enough to be executed on a standard IoT device?
\end{itemize}

\begin{table}[t]
 \centering
 \caption{The STL formulas used in the experiments.}%
 \label{table:specifications}
 \scriptsize
 \begin{tabular}{c c c c}
  \toprule
  $\BloodGlucose_1$ & $ \BoxOp{[100,700]} (\mathit{glucose} \geq 70)$ & & \\
  $\BloodGlucose_2$ & $ \BoxOp{[100,700]} (\mathit{glucose} < 350)$ & & \\
  $\BloodGlucose_4$ & $ \BoxOp{[600,700]} (\mathit{glucose} < 200)$ & & \\
  $\BloodGlucose_5$ & $ \neg\DiaOp{[200,600]}\BoxOp{[0,180]} (\mathit{glucose} \geq 240)$ & & \\
  $\BloodGlucose_6$ & $ \neg\DiaOp{[200,600]}\BoxOp{[0,180]} (\mathit{glucose} < 70)$ & & \\
  \bottomrule
 \end{tabular}
 \hfil
 \begin{tabular}{c c c c}
 \toprule
  $\BloodGlucose_7$ & $ \BoxOp{} (\mathit{glucose} \geq 70 \land \mathit{glucose} < 180)$ & & \\
  $\BloodGlucose_8$ & $ \BoxOp{} (\Delta \mathit{glucose} \geq -5 \land \Delta \mathit{glucose} < 3)$ & & \\
  $\BloodGlucose_{10}$ & $ \BoxOp{} (\mathit{glucose} < 60 \imply \DiaOp{[0,25]} \mathit{glucose} \geq 60)$ & & \\
  $\BloodGlucose_{11}$ & $ \BoxOp{} (\mathit{glucose} > 200 \imply \DiaOp{[0,25]} \mathit{glucose} < 200)$ & & \\
  \midrule
  $\RSS$ & $\BoxOp{} (\Safe \land \NextOp{} \lnot \Safe \imply \NextOp{} (\fmlPreBr \land \fmlBrake))$ \\
  \bottomrule
 \end{tabular}
\end{table}
\begin{figure}[t]
 \centering
 \begin{gather*}
  \fmlPreBr \equiv \Safe \TRelease{[0, \reactionDelayInStep)} (\AEgoMaxAcc \land \APrecMaxBr) \qquad
  \fmlBrake \equiv \Safe \TRelease{[\reactionDelayInStep, +\infty)} (\AEgoMinBr \land \APrecMaxBr)\\
  \Safe \equiv (\xOf{\egoIndex} - \xOf{\precIndex} < w \land \xOf{\precIndex} - \xOf{\egoIndex}< w) \imply \yOf{\egoIndex} < \yOf{\precIndex} \land \yOf{\egoIndex} < \yOf{\precIndex} - \dRSS\\
   \AEgoMaxAcc \equiv \ayOf{\egoIndex} < \ayMaxAcc \quad
   \AEgoMinBr \equiv \ayOf{\egoIndex} < -\ayMinBr\quad
  \APrecMaxBr \equiv \ayOf{\precIndex} \geq -\ayMaxBr\\
  \dRSS = \dEgoPreBr + \dEgoBrake - \dPrecBrake \qquad
  \dEgoPreBr = \vyOf{\egoIndex} \reactionDelay + 0.5 \ayMaxAcc {\reactionDelay}^{2}\\
  \dEgoBrake = {(\vyOf{\egoIndex} + \reactionDelay \ayMaxAcc)}^{2}/({2 \ayMinBr}) \qquad
  \dPrecBrake = \vyOf{\precIndex}^2/(2 \ayMinBr)
 \end{gather*}
 \caption{The auxiliary predicates and formulas in \RSS{}. The constants are shown in\LongVersion{ \cref{section:detail_rss_formula}}\ShortVersion{ Appendix B.2 of~\cite{DBLP:journals/corr/abs-2405-16767}}.}%
 \label{figure:auxiliary_RSS}
\end{figure}
\begin{table}[tp]
 \captionof{table}{Summary of the benchmarks. The columns $|\Loc|$ and $|\reverse{\Loc}|$ show the number of states of $\Monitor_{\fml}$ and $\reverse{\Monitor}_{\fml}$, respectively. The columns $|\word|$ show the length of the monitored log. Cells labeled \TIMEOUT{} indicate DFA construction was halted due to memory exhaustion.}%
 \label{table:benchmark_summary}
 \scriptsize
 \begin{tabular}{lrrrrrcccccccc}
  \toprule
  & \multicolumn{2}{c}{\ourTool{}} & \multicolumn{2}{c}{\homfa{}} &  \multirow{2}{*}{$|\word|$}\\
  & $|\Loc|$ & $|\reverse{\Loc}|$ & $|\Loc|$ & $|\reverse{\Loc}|$ &  &\\
  \midrule
  $\BloodGlucose_{1}$ & 703 & 172,402 & 10,524 & \TIMEOUT{} & 721 \\
  $\BloodGlucose_{2}$ & 703 & 172,402 & 11,126 & \TIMEOUT{} & 721 \\
  $\BloodGlucose_{4}$ & 703 & \TIMEOUT{} & 7026 & \TIMEOUT{} & 721 \\
  $\BloodGlucose_{5}$ & 72,603 & \TIMEOUT{} & \TIMEOUT{} & \TIMEOUT{} & 721 \\
  $\BloodGlucose_{6}$ & 72,603 & \TIMEOUT{} & \TIMEOUT{} & \TIMEOUT{} & 721 \\
  \bottomrule
 \end{tabular}
 \hfill
 \begin{tabular}{lrrrrrcccccccc}
  \toprule
  & \multicolumn{2}{c}{\ourTool{}} & \multicolumn{2}{c}{\homfa{}} &  \multirow{2}{*}{$|\word|$}\\
  & $|\Loc|$ & $|\reverse{\Loc}|$ & $|\Loc|$ & $|\reverse{\Loc}|$ &  &\\
  \midrule
  $\BloodGlucose_{7}$ & 3 & 3 & 21 & 20 & 10,081 \\
  $\BloodGlucose_{8}$ & 5 & 5 & --- & --- & 10,081 \\
  $\BloodGlucose_{10}$ & 27 & 27 & 237 & 237 & 10,081 \\
  $\BloodGlucose_{11}$ & 27 & 27 & 390 & 390 & 10,081 \\
  \midrule
  \RSS{} & 179 & 218 & --- & --- & 49 \\
  \bottomrule
 \end{tabular}
\end{table}
\subsection{Benchmarks}\label{section:benchmarks}
For the experiments addressing RQ1--3,
we used two benchmarks, \BloodGlucose{} and \RSS{} from practical applications.
\cref{table:specifications} shows the STL formulas we used, where the auxiliary predicates and formulas in \RSS{} are shown in \cref{figure:auxiliary_RSS}.
\cref{table:benchmark_summary} shows the size of the DFAs and the length of the monitored signals.

\BloodGlucose{} is taken from~\cite{DBLP:conf/cav/BannoMMBWS22}.
It is a benchmark for monitoring of blood glucose levels of type 1 diabetes patients.
The monitored signals are one-dimensional with $\Var = \{ \mathit{glucose} \}$.
The term $\Delta \mathit{glucose}$ in $\BloodGlucose_8$ represents the difference between two consecutive blood glucose levels, computed by CKKS in our algorithm.
We used simglucose~\cite{simglucose} to generate the monitored signals.
STL formulas $\BloodGlucose_1$--$\BloodGlucose_6$ originate from~\cite{DBLP:conf/rv/CameronFMS15}.
STL formulas $\BloodGlucose_7$--$\BloodGlucose_{12}$ originate from~\cite{DBLP:conf/iotdi/YoungCGPF18}.

\RSS{} is a benchmark for monitoring driving behaviors of vehicles against the \emph{Responsibility-Sensitive Safety (RSS)} model~\cite{DBLP:journals/corr/abs-1708-06374}.
The monitored signal comprises eight dimensions: lateral position, longitudinal position, longitudinal velocity, and longitudinal acceleration of both the ego and the preceding vehicle, represented as $\Var = \{ \xOf{i}, \yOf{i}, \vyOf{i}, \ayOf{i} \mid i \in \{\mathrm{ego}, \mathrm{prec}\} \}$.
We generated the monitored signal using an unpublished 2D driving simulator\LongVersion{ to visualize driving scenarios}.
We used a variant of an STL formula taken from~\cite{DBLP:conf/memocode/HekmatnejadYDAS19} encoding the RSS rule.

\subsection{Experiments}\label{section:experiment_setting}

To answer RQ1--3, we measured the time required for monitoring encrypted logs using \BloodGlucose{} and \RSS{}.
For RQ2, we compared \ourTool{} with and without the optimization (\Optimized{} and \Naive{}) in \cref{section:protocol:ckks_to_tfhe}.
For RQ3, we compared \Optimized{} with \homfa{}~\cite{DBLP:conf/cav/BannoMMBWS22}, a tool for online oblivious \emph{LTL} monitoring only with TFHE.\@
To answer RQ4, we measured the time\LongVersion{ taken} to
\begin{ienumeration}
 \item encrypt random numbers to RLWE ciphertexts with either a public or private key within\LongVersion{ the} CKKS\LongVersion{ scheme} and
 \item decrypt random LWE ciphertexts within\LongVersion{ the} TFHE\LongVersion{ scheme}.
\end{ienumeration}
For the experiments addressing RQ4, we used randomly generated values because the values do not affect the time for encryption and decryption.

For the experiments addressing RQ1--3,
we ran three approaches (\Optimized{}, \Naive{}, and \homfa{}) using two DFA execution algorithms  (\REVERSE{} and \BLOCK{}) proposed in~\cite{DBLP:conf/cav/BannoMMBWS22}.
Thus, we have six configurations in total.
We ran each configuration five times and took the mean runtime of these executions.

In \Naive{}, we apply $\FHomDecomp$ to 64bits, while in \Optimized{}, we apply $\FHomDecomp$ to the first 24bits.
For \ourTool{}, we used a block size (in \BLOCK{}) of $\BlockSize = 1$ and a bootstrapping interval (in \REVERSE{}) of $\BootInterval = 200$, \ie{}
\BLOCK{} returns the results for each input and
bootstrapping is performed every time 200 ciphertexts are processed by \REVERSE{}.
For \homfa{}, we used the same parameters as in~\cite{DBLP:conf/cav/BannoMMBWS22}\footnote{Our definition of block size $\BlockSize$ is slightly different from~\cite{DBLP:conf/cav/BannoMMBWS22}. In their definition, the block size is $9$ for \homfa{}, which is the same as the experiments in~\cite{DBLP:conf/cav/BannoMMBWS22}}:
$\BlockSize = 1$ and $\BootInterval = 30000$.
\homfa{} is built with Spot~2.9.7~\cite{DBLP:conf/cav/Duret-LutzRCRAS22}, which is also the same as~\cite{DBLP:conf/cav/BannoMMBWS22}.
\ourTool{} requires more frequent bootstrapping because the ciphertext obtained by scheme switching has larger noise.
The choice of $\BlockSize = 1$ is for the consistency with~\cite{DBLP:conf/cav/BannoMMBWS22}.
We used the default security parameters of SEAL and TFHEpp, which satisfy 128-bit security.
See\LongVersion{ \cref{section:parameters}}\ShortVersion{ Appendix A of~\cite{DBLP:journals/corr/abs-2405-16767}} for the concrete parameters.

We conducted the experiments for RQ1--3 on an AWS EC2 c6i.4xlarge instance (16 vCPU, 32GB RAM) running Ubuntu 22.04.
For the experiments for RQ4, we used two single-board computers (SBCs): one with and one without an Advanced Encryption Standard (AES)~\cite{daemen1999aes} hardware accelerator.
Specifically, we used the Raspberry Pi 4 model B (ARM Cortex-A72 \emph{without} a hardware AES accelerator) with 4 GiB RAM, running Ubuntu 23.04, and the NanoPi R6S with a Rockchip RK3588S (ARM Cortex-A76 and Cortex-A55 \emph{with} a hardware AES accelerator) and 8 GiB RAM, running Ubuntu 22.04.2 LTS.\@

In \homfa{}, each blood glucose level was encoded using nine Boolean values via bit encoding, which is the same as~\cite{DBLP:conf/cav/BannoMMBWS22}.
For $\BloodGlucose_8$ and $\RSS$, we could not run \homfa{} because the STL formulas include arithmetic operations\LongVersion{ in the predicates}.

\begin{table*}[tb]
 \centering
 \caption{Summary of the runtime of \ourTool{} with and without the optimization in \cref{section:protocol:ckks_to_tfhe}. The meaning of \TIMEOUT{} is the same as \cref{table:benchmark_summary}. For each STL formula, the fastest configuration is highlighted.}%
 \label{table:result_summary}
 \scriptsize
 \begin{tabular}{lrrrrrrrr}
  \toprule
  & \multicolumn{4}{c}{Runtime (sec.)} & \multicolumn{4}{c}{Runtime/valuation (sec.)} \\
  & \multicolumn{2}{c}{\Optimized{}} & \multicolumn{2}{c}{\Naive{}} & \multicolumn{2}{c}{\Optimized{}} & \multicolumn{2}{c}{\Naive{}} \\
  & \BLOCK{} & \REVERSE{} & \BLOCK{} & \REVERSE{} & \BLOCK{} & \REVERSE{} & \BLOCK{} & \REVERSE{} \\
  \midrule
  $\BloodGlucose_{1}$ & \tbcolor{} 2.33e+02 & 4.75e+02 & 3.60e+02 & 6.13e+02 & \tbcolor{} 3.23e-01 & 6.59e-01 & 5.00e-01 & 8.50e-01 \\
  $\BloodGlucose_{2}$ & \tbcolor{} 2.29e+02 & 4.76e+02 & 3.73e+02 & 6.18e+02 & \tbcolor{} 3.18e-01 & 6.60e-01 & 5.17e-01 & 8.58e-01 \\
  $\BloodGlucose_{4}$ & \tbcolor{} 1.83e+02 & \TIMEOUT{} & 3.28e+02 & \TIMEOUT{} & \tbcolor{} 2.54e-01 & \TIMEOUT{} & 4.55e-01 & \TIMEOUT{} \\
  $\BloodGlucose_{5}$ & \tbcolor{} 3.30e+02 & \TIMEOUT{} & 4.59e+02 & \TIMEOUT{} & \tbcolor{} 4.58e-01 & \TIMEOUT{} & 6.36e-01 & \TIMEOUT{} \\
  $\BloodGlucose_{6}$ & \tbcolor{} 3.74e+02 & \TIMEOUT{} & 4.68e+02 & \TIMEOUT{} & \tbcolor{} 5.19e-01 & \TIMEOUT{} & 6.50e-01 & \TIMEOUT{} \\
  \midrule
  $\BloodGlucose_{7}$ & 3.96e+03 & \tbcolor{} 2.93e+03 & 5.90e+03 & 5.06e+03 & 3.93e-01 & \tbcolor{} 2.90e-01 & 5.85e-01 & 5.01e-01 \\
  $\BloodGlucose_{8}$ & 3.87e+03 & \tbcolor{} 3.03e+03 & 6.05e+03 & 5.09e+03 & 3.84e-01 & \tbcolor{} 3.00e-01 & 6.00e-01 & 5.05e-01 \\
  $\BloodGlucose_{10}$ & 3.66e+03 & \tbcolor{} 2.50e+03 & 5.43e+03 & 4.26e+03 & 3.63e-01 & \tbcolor{} 2.48e-01 & 5.39e-01 & 4.23e-01 \\
  $\BloodGlucose_{11}$ & 3.49e+03 & \tbcolor{} 2.65e+03 & 5.40e+03 & 4.23e+03 & 3.46e-01 & \tbcolor{} 2.62e-01 & 5.36e-01 & 4.19e-01 \\
  \midrule
  \RSS{} & 2.79e+01 & \tbcolor{} 2.50e+01 & 4.62e+01 & 3.87e+01 & 5.69e-01 & \tbcolor{} 5.11e-01 & 9.43e-01 & 7.89e-01 \\
  \bottomrule
 \end{tabular}
\end{table*}
\begin{table*}[tbp]
 \centering
 \caption{Summary of the runtime of \homfa{}. The meaning of \TIMEOUT{} is the same as \cref{table:benchmark_summary}. The STL formulas \homfa{} could not handle are omitted.}%
 \label{table:result_homfa_summary}
 \scriptsize
 \begin{tabular}{lrrrr}
  \toprule
  & \multicolumn{2}{c}{Runtime (sec.)} & \multicolumn{2}{c}{Runtime/val. (sec.)} \\
  & \BLOCK{} & \REVERSE{} & \BLOCK{} & \REVERSE{} \\
  \midrule
  $\BloodGlucose_{1}$ & 7.33e+01 & \TIMEOUT{} & 1.02e-01 & \TIMEOUT{} \\
  $\BloodGlucose_{2}$ & 7.36e+01 & \TIMEOUT{} & 1.02e-01 & \TIMEOUT{} \\
  $\BloodGlucose_{4}$ & 2.05e+01 & \TIMEOUT{} & 2.84e-02 & \TIMEOUT{} \\
  \bottomrule
 \end{tabular}
 \hfill
 \begin{tabular}{lrrrr}
  \toprule
  & \multicolumn{2}{c}{Runtime (sec.)} & \multicolumn{2}{c}{Runtime/val. (sec.)} \\
  & \BLOCK{} & \REVERSE{} & \BLOCK{} & \REVERSE{} \\
  \midrule
  $\BloodGlucose_{7}$ & 9.58e+02 & 8.83e+00 & 9.50e-02 & 8.76e-04 \\
  $\BloodGlucose_{10}$ & 1.12e+03 & 5.59e+01 & 1.11e-01 & 5.54e-03 \\
  $\BloodGlucose_{11}$ & 1.15e+03 & 8.91e+01 & 1.14e-01 & 8.84e-03 \\
  \bottomrule
 \end{tabular}
\end{table*}

\cref{table:result_summary,table:result_homfa_summary} show the mean runtime of \ourTool{} and \homfa{}.
\cref{table:enc_dec_result} shows the mean runtime to encrypt and decrypt ciphertexts on SBCs.

\subsection{RQ1: Performance on practical benchmarks}\label{section:practicality}

In the ``\Optimized{}'' block of ``Runtime per valuation (sec.)'' column of \cref{table:result_summary},
we observe that for any formula in \BloodGlucose{}, and for any algorithm,
the mean runtime per signal valuation is at most 700 milliseconds when DFA construction was successful and the optimization is used. %
Furthermore, for any formula in \BloodGlucose{},
the mean runtime per signal valuation is less than 500 milliseconds if an appropriate algorithm is chosen.
This is much faster than the typical sampling intervals for blood glucose levels, \eg{} 5 minutes for Dexcom G6~\cite{klyve2023algorithm}.

In the ``\Optimized{}'' block of ``Runtime per valuation (sec.)'' column of \cref{table:result_summary},
we observe that for \RSS{},
the mean runtime per signal valuation is less than 550 milliseconds for \REVERSE{} and less than 650 milliseconds for \BLOCK{} if the optimization is used.
Since this closely aligns with the reaction time of human drivers~\cite{10.5555/1295645.1295744},
a delay of around 550 milliseconds is likely acceptable for driver alert systems by sending precautions.
Overall, we answer RQ1 as follows.

\rqanswer{RQ1}{\ourTool{} can monitor each signal valuation in less than 550 milliseconds if an appropriate method is used. This throughput is sufficient for monitoring blood glucose levels. Additionally, it is likely fast enough for monitoring driving behaviors as part of a driver alert system.}

\subsection{RQ2: Acceleration by the optimized scheme switching}\label{section:domain_knowledge}

In the ``Runtime per valuation (sec.)'' column of \cref{table:result_summary},
we observe that the performance improvement (per signal valuation) by our optimization in \cref{section:protocol:ckks_to_tfhe} was about 150 milliseconds for \BloodGlucose{} and about 250 milliseconds for \RSS{}.
The runtime of \Optimized{} was about $70\%$ of that of \Naive{}.
This large reduction is because scheme switching is the dominant bottleneck among the overall process.
Overall, we answer RQ2 as follows.

\rqanswer{RQ2}{The optimization we introduced in \cref{section:protocol:ckks_to_tfhe} reduces the overall execution time about $30\%$ from the naive scheme switching.}

\subsection{RQ3: Comparison with purely TFHE-based approach}\label{section:comparison}

In the ``Runtime (sec.)'' columns of \cref{table:result_summary,table:result_homfa_summary},
we observe that for benchmarks where both \ourTool{} and \homfa{} work (\eg{} $\BloodGlucose_1$, $\BloodGlucose_2$, $\BloodGlucose_4$, $\BloodGlucose_7$, $\BloodGlucose_{10}$, and $\BloodGlucose_{11}$ for \BLOCK{}), \homfa{} is more efficient due to the computational cost of scheme switching.
In contrast, we observe that \homfa{} cannot handle some STL formulas due to excessive memory consumption (\eg{} $\BloodGlucose_5$ and $\BloodGlucose_6$ for \BLOCK{}).
This is due to the difference in the DFA size representing the same specification (\cref{table:benchmark_summary}).
In \homfa{}, each signal valuation is encoded by nine ciphertexts via bit encoding.
In \ourTool{}, we have one ciphertext for each atomic proposition, which is at most two in \BloodGlucose{}.
This distinction results in the DFAs for \ourTool{} being much smaller than those for \homfa{}.
Overall, we answer RQ3 as follows.

\rqanswer{RQ3}{\homfa{} is faster than \ourTool{} for the specifications that both can handle. For specifications with a huge DFA encoding, only \ourTool{} works due to the differences in DFA construction.}

\subsection{RQ4: Computational demand of the client}\label{section:client_demand}
\begin{table*}[t]
 \centering
 \caption{Mean runtime to encrypt RLWE ciphertexts within\LongVersion{ the} CKKS\LongVersion{ scheme} and decrypt LWE ciphertexts within \LongVersion{the} TFHE\LongVersion{ scheme} on SBCs.}%
 \label{table:enc_dec_result}
 \scriptsize
 \begin{tabular}{c c c c}
  \toprule
  & Enc. w/ public key &  Enc. w/ private key &  Decryption \\
  & [ms/value] & [ms/value] & [ms/ciphertext] \\
  \midrule
  NanoPi R6S (\emph{w/} AES accelerator) & 6.82 & 2.21 & $1.17 \times 10^{-3}$ \\
  Raspberry Pi 4 (\emph{w/o} AES accelerator) & 12.7 & 4.44 & $1.72 \times 10^{-3}$ \\
  \bottomrule
 \end{tabular}
\end{table*}

\cref{table:enc_dec_result} summarizes the mean runtime of encryption and decryption on SBCs.
We observe that both the encryption and decryption processes consume considerably less time than the throughput of monitoring shown in the ``Runtime per valuation (sec.)''  column of \cref{table:result_summary}.
Therefore, we answer RQ4 as follows.

\rqanswer{RQ4}{The client's computational demand in the proposed protocol is sufficiently low for standard IoT devices.}

Furthermore, \cref{table:enc_dec_result} shows that the encryption and decryption processes are faster on the NanoPi R6S than on the Raspberry Pi 4.
This is likely due to the efficiency of the random number generation, which is enhanced by the hardware AES accelerator, as reported in~\cite{DBLP:conf/cav/BannoMMBWS22} for TFHE.\@

\begin{ShortVersionBlock}  
\section{Related works}\label{sec:related_work}
 \begin{table}[tbp]
 \centering
 \scriptsize
 \caption{Comparison of FHE-based monitoring methods.}%
 \label{table:related_methods}
 \begin{tabular}{c c c c}
  \toprule
  & Used FHE Scheme & Arithmetic operations & Specification is secret\\\midrule
  \textbf{Ours} & CKKS~\cite{DBLP:conf/asiacrypt/CheonKKS17} and TFHE~\cite{DBLP:journals/joc/ChillottiGGI20} & \goodFeature{Yes} & \goodFeature{Yes} \\
  \cite{DBLP:conf/cav/BannoMMBWS22}  & TFHE~\cite{DBLP:journals/joc/ChillottiGGI20} & \badFeature{No} & \goodFeature{Yes} \\ %
  \cite{DBLP:conf/trustbus/TriakosiaRTTSF22} & CKKS~\cite{DBLP:conf/asiacrypt/CheonKKS17} & \goodFeature{Yes} & \badFeature{No}\\ %
  \bottomrule
 \end{tabular}
 \end{table}

 \cref{table:related_methods} summarizes FHE-based monitoring methods.
 As previously mentioned, our method is based on~\cite{DBLP:conf/cav/BannoMMBWS22} and handles arithmetic operations by bridging the CKKS and TFHE schemes.
 Triakosia et al.~\cite{DBLP:conf/trustbus/TriakosiaRTTSF22} proposed a method for oblivious monitoring of manufacturing quality measures with CKKS.\@
 The approach in~\cite{DBLP:conf/trustbus/TriakosiaRTTSF22} is collaborative:
 the server conducts polynomial operations using CKKS, while
 the client conducts non-polynomial operations (\eg{} branching) without using FHE techniques.
 This is done by
 \begin{ienumeration}
 \item decrypting the ciphertexts sent from the server,
 \item conducting non-polynomial operations over plaintexts,
 \item encrypting the result, and
 \item sending it back to the server.
 \end{ienumeration}
 This collaborative approach inherently allows the client to access the monitored specification.
 In contrast, our monitoring algorithm runs entirely on the server, thereby ensuring the specification remains confidential and not exposed to the client.

\end{ShortVersionBlock}

\section{Conclusions and future work}\label{section:conclusion}

Combining two FHE schemes, we proposed a protocol for online oblivious monitoring of safety STL formulas with arithmetic operations.
We evaluated the proposed approach by monitoring blood glucose levels and vehicles' behavior\LongVersion{ against the RSS rules}.
The experimental results suggest the practical relevance of our protocol.

 Possible future directions include extending the protocol to handle more general signals, \eg{} mixed signals~\cite{DBLP:conf/formats/HavlicekLMN10}.
 \LongVersion{Our monitoring algorithm can be seen as an algorithm to run symbolic automata~\cite{DBLP:conf/cav/DAntoniV17} over real vectors and linear constraints.
 An extension for other predicates is also a future direction.}
 \LongVersion{Another direction is the extension for robust STL monitoring~\cite{FP09}.}

\subsubsection*{Acknowledgments}

This work is partially supported
    by
    JST PRESTO Grant No.~JPMJPR22CA,
    JSPS KAKENHI Grant No.~22K17873 and 23KJ1319, and
    JST CREST Grant No.~JPMJCR19K5, JPMJCR2012, and JPMJCR21M3.

\ifdefined\VersionLong%
\newcommand{\LNCS}{Lecture Notes in Computer Science}
\else
\newcommand{\LNCS}{LNCS}
\fi

\clearpage
\bibliography{dblp_refs}

\begin{thebibliography}{10}
\providecommand{\url}[1]{\texttt{#1}}
\providecommand{\urlprefix}{URL }
\providecommand{\doi}[1]{https://doi.org/#1}

\bibitem{DBLP:conf/cav/BannoMMBWS22}
Banno, R., Matsuoka, K., Matsumoto, N., Bian, S., Waga, M., Suenaga, K.:
  Oblivious online monitoring for safety {LTL} specification via fully
  homomorphic encryption. In: Shoham, S., Vizel, Y. (eds.) Computer Aided
  Verification - 34th International Conference, {CAV} 2022, Proceedings, Part
  {I}. \LNCS, vol. 13371, pp. 447--468. Springer (2022)

\bibitem{DBLP:journals/tosem/BauerLS11}
Bauer, A., Leucker, M., Schallhart, C.: Runtime verification for {LTL} and
  {TLTL}. {ACM} Trans. Softw. Eng. Methodol.  \textbf{20}(4),  14:1--14:64
  (2011)

\bibitem{DBLP:conf/ccs/00010PMZJG23}
Bian, S., Zhang, Z., Pan, H., Mao, R., Zhao, Z., Jin, Y., Guan, Z.: {HE3DB:} an
  efficient and elastic encrypted database via arithmetic-and-logic fully
  homomorphic encryption. In: Meng, W., Jensen, C.D., Cremers, C., Kirda, E.
  (eds.) Proceedings of the 2023 {ACM} {SIGSAC} Conference on Computer and
  Communications Security, {CCS} 2023, Copenhagen, Denmark, November 26-30,
  2023. pp. 2930--2944. {ACM} (2023)

\bibitem{DBLP:journals/jmc/BouraGGJ20}
Boura, C., Gama, N., Georgieva, M., Jetchev, D.: {CHIMERA:} combining
  ring-lwe-based fully homomorphic encryption schemes. J. Math. Cryptol.
  \textbf{14}(1),  316--338 (2020)

\bibitem{DBLP:conf/rv/CameronFMS15}
Cameron, F., Fainekos, G., Maahs, D.M., Sankaranarayanan, S.: Towards a
  verified artificial pancreas: Challenges and solutions for runtime
  verification. In: Bartocci, E., Majumdar, R. (eds.) Runtime Verification -
  6th International Conference, {RV} 2015. Proceedings. \LNCS, vol.~9333, pp.
  3--17. Springer (2015)

\bibitem{DBLP:conf/eurocrypt/CheonHKKS18}
Cheon, J.H., Han, K., Kim, A., Kim, M., Song, Y.: Bootstrapping for approximate
  homomorphic encryption. In: Nielsen, J.B., Rijmen, V. (eds.) Advances in
  Cryptology - {EUROCRYPT} 2018 - 37th Annual International Conference on the
  Theory and Applications of Cryptographic Techniques, Proceedings, Part {I}.
  \LNCS, vol. 10820, pp. 360--384. Springer (2018)

\bibitem{DBLP:conf/asiacrypt/CheonKKS17}
Cheon, J.H., Kim, A., Kim, M., Song, Y.S.: Homomorphic encryption for
  arithmetic of approximate numbers. In: Takagi, T., Peyrin, T. (eds.) Advances
  in Cryptology - {ASIACRYPT} 2017 - 23rd International Conference on the
  Theory and Applications of Cryptology and Information Security, Proceedings,
  Part {I}. \LNCS, vol. 10624, pp. 409--437. Springer (2017)

\bibitem{DBLP:journals/joc/ChillottiGGI20}
Chillotti, I., Gama, N., Georgieva, M., Izabach{\`{e}}ne, M.: {TFHE:} fast
  fully homomorphic encryption over the torus. J. Cryptol.  \textbf{33}(1),
  34--91 (2020)

\bibitem{daemen1999aes}
Daemen, J., Rijmen, V.: Aes proposal: Rijndael  (1999)

\bibitem{DBLP:conf/cav/DAntoniV17}
D'Antoni, L., Veanes, M.: The power of symbolic automata and transducers. In:
  Majumdar, R., Kuncak, V. (eds.) Computer Aided Verification - 29th
  International Conference, {CAV} 2017, Proceedings, Part {I}. \LNCS, vol.
  10426, pp. 47--67. Springer (2017)

\bibitem{DBLP:conf/cav/Duret-LutzRCRAS22}
Duret{-}Lutz, A., Renault, E., Colange, M., Renkin, F., Aisse, A.G.,
  Schlehuber{-}Caissier, P., Medioni, T., Martin, A., Dubois, J., Gillard, C.,
  Lauko, H.: From spot 2.0 to spot 2.10: What's new? In: Shoham, S., Vizel, Y.
  (eds.) Computer Aided Verification - 34th International Conference, {CAV}
  2022, Proceedings, Part {II}. \LNCS, vol. 13372, pp. 174--187. Springer
  (2022)

\bibitem{FP09}
Fainekos, G.E., Pappas, G.J.: Robustness of temporal logic specifications for
  continuous-time signals. Theor. Comput. Sci.  \textbf{410}(42),  4262--4291
  (2009)

\bibitem{DBLP:conf/stoc/Gentry09}
Gentry, C.: Fully homomorphic encryption using ideal lattices. In:
  Mitzenmacher, M. (ed.) Proceedings of the 41st Annual {ACM} Symposium on
  Theory of Computing, {STOC} 2009. pp. 169--178. {ACM} (2009)

\bibitem{DBLP:conf/crypto/GentrySW13}
Gentry, C., Sahai, A., Waters, B.: Homomorphic encryption from learning with
  errors: Conceptually-simpler, asymptotically-faster, attribute-based. In:
  Canetti, R., Garay, J.A. (eds.) Advances in Cryptology - {CRYPTO} 2013 - 33rd
  Annual Cryptology Conference. Proceedings, Part {I}. \LNCS, vol.~8042, pp.
  75--92. Springer (2013)

\bibitem{DBLP:conf/formats/HavlicekLMN10}
Havlicek, J., Little, S., Maler, O., Nickovic, D.: Property-based monitoring of
  analog and mixed-signal systems. In: Chatterjee, K., Henzinger, T.A. (eds.)
  Formal Modeling and Analysis of Timed Systems - 8th International Conference,
  {FORMATS} 2010. Proceedings. \LNCS, vol.~6246, pp. 23--24. Springer (2010)

\bibitem{DBLP:conf/memocode/HekmatnejadYDAS19}
Hekmatnejad, M., Yaghoubi, S., Dokhanchi, A., Amor, H.B., Shrivastava, A.,
  Karam, L.J., Fainekos, G.: Encoding and monitoring responsibility sensitive
  safety rules for automated vehicles in signal temporal logic. In: Roop, P.S.,
  Zhan, N., Gao, S., Nuzzo, P. (eds.) Proceedings of the 17th {ACM-IEEE}
  International Conference on Formal Methods and Models for System Design,
  {MEMOCODE} 2019. pp. 6:1--6:11. {ACM} (2019)

\bibitem{klyve2023algorithm}
Klyve, D., Currie, K., Anderson~Jr, J.H., Ward, C., Schwarz, D., Shelton, B.:
  Algorithm refinement in the non-invasive detection of blood glucose via
  bio-rfid™ technology1. medRxiv pp. 2023--05 (2023)

\bibitem{DBLP:journals/fmsd/KupfermanV01}
Kupferman, O., Vardi, M.Y.: Model checking of safety properties. Formal Methods
  Syst. Des.  \textbf{19}(3),  291--314 (2001)

\bibitem{DBLP:conf/asiacrypt/LiuMP22}
Liu, Z., Micciancio, D., Polyakov, Y.: Large-precision homomorphic sign
  evaluation using {FHEW/TFHE} bootstrapping. In: Agrawal, S., Lin, D. (eds.)
  Advances in Cryptology - {ASIACRYPT} 2022 - 28th International Conference on
  the Theory and Application of Cryptology and Information Security, Taipei,
  Taiwan, December 5-9, 2022, Proceedings, Part {II}. Lecture Notes in Computer
  Science, vol. 13792, pp. 130--160. Springer (2022)

\bibitem{DBLP:conf/sp/LuHHMQ21}
Lu, W., Huang, Z., Hong, C., Ma, Y., Qu, H.: {PEGASUS:} bridging polynomial and
  non-polynomial evaluations in homomorphic encryption. In: 42nd {IEEE}
  Symposium on Security and Privacy, {SP} 2021. pp. 1057--1073. {IEEE} (2021)

\bibitem{DBLP:journals/jacm/LyubashevskyPR13}
Lyubashevsky, V., Peikert, C., Regev, O.: On ideal lattices and learning with
  errors over rings. J. {ACM}  \textbf{60}(6),  43:1--43:35 (2013)

\bibitem{DBLP:journals/tches/MaHWZW24}
Ma, S., Huang, T., Wang, A., Zhou, Q., Wang, X.: Fast and accurate: Efficient
  full-domain functional bootstrap and digit decomposition for homomorphic
  computation. {IACR} Trans. Cryptogr. Hardw. Embed. Syst.  \textbf{2024}(1),
  592--616 (2024)

\bibitem{DBLP:conf/formats/MalerN04}
Maler, O., Nickovic, D.: Monitoring temporal properties of continuous signals.
  In: Lakhnech, Y., Yovine, S. (eds.) Formal Techniques, Modelling and Analysis
  of Timed and Fault-Tolerant Systems, Joint International Conferences on
  Formal Modelling and Analysis of Timed Systems, {FORMATS} 2004 and Formal
  Techniques in Real-Time and Fault-Tolerant Systems, {FTRTFT} 2004,
  Proceedings. \LNCS, vol.~3253, pp. 152--166. Springer (2004)

\bibitem{DBLP:conf/uss/MatsuokaBMS021}
Matsuoka, K., Banno, R., Matsumoto, N., Sato, T., Bian, S.: Virtual secure
  platform: {A} five-stage pipeline processor over {TFHE}. In: Bailey, M.,
  Greenstadt, R. (eds.) 30th {USENIX} Security Symposium, {USENIX} Security
  2021. pp. 4007--4024. {USENIX} Association (2021)

\bibitem{DBLP:conf/focs/Pnueli77}
Pnueli, A.: The temporal logic of programs. In: 18th Annual Symposium on
  Foundations of Computer Science, 1977. pp. 46--57. {IEEE} Computer Society
  (1977)

\bibitem{DBLP:journals/jacm/Regev09}
Regev, O.: On lattices, learning with errors, random linear codes, and
  cryptography. J. {ACM}  \textbf{56}(6),  34:1--34:40 (2009)

\bibitem{DBLP:journals/pvldb/RenSG0000LZ22}
Ren, X., Su, L., Gu, Z., Wang, S., Li, F., Xie, Y., Bian, S., Li, C., Zhang,
  F.: {HEDA:} multi-attribute unbounded aggregation over homomorphically
  encrypted database. Proc. {VLDB} Endow.  \textbf{16}(4),  601--614 (2022)

\bibitem{rivest1978data}
Rivest, R.L., Adleman, L., Dertouzos, M.L., et~al.: On data banks and privacy
  homomorphisms. Foundations of secure computation  \textbf{4}(11),  169--180
  (1978)

\bibitem{sealcrypto}
{M}icrosoft {SEAL} (release 4.1). \url{https://github.com/Microsoft/SEAL} (Jan
  2023), microsoft Research, Redmond, WA.

\bibitem{DBLP:journals/corr/abs-1708-06374}
Shalev{-}Shwartz, S., Shammah, S., Shashua, A.: On a formal model of safe and
  scalable self-driving cars. CoRR  \textbf{abs/1708.06374} (2017)

\bibitem{DBLP:conf/trustbus/TriakosiaRTTSF22}
Triakosia, A., Rizomiliotis, P., Tserpes, K., Tonelli, C., Senni, V., Federici,
  F.: Homomorphic encryption in manufacturing compliance checks. In: Katsikas,
  S.K., Furnell, S. (eds.) Trust, Privacy and Security in Digital Business -
  19th International Conference, TrustBus 2022, Proceedings. \LNCS, vol. 13582,
  pp. 81--95. Springer (2022)

\bibitem{simglucose}
Xie, J.: Simglucose v0.2.1. \url{https://github.com/jxx123/simglucose} (2018),
  accessed: 2023-05-01

\bibitem{DBLP:conf/iotdi/YoungCGPF18}
Young, W., Corbett, J., Gerber, M.S., Patek, S., Feng, L.: {DAMON:} {A} data
  authenticity monitoring system for diabetes management. In: 2018 {IEEE/ACM}
  Third International Conference on Internet-of-Things Design and
  Implementation, IoTDI 2018. pp. 25--36. {IEEE} Computer Society (2018)

\bibitem{10.5555/1295645.1295744}
Zhang, X., Bham, G.H.: Estimation of driver reaction time from detailed vehicle
  trajectory data. In: Proceedings of the 18th Conference on Proceedings of the
  18th IASTED International Conference: Modelling and Simulation. p. 574–579.
  MOAS'07, ACTA Press, USA (2007)

\end{thebibliography}
\begin{LongVersionBlock}
\clearpage
\appendix

\begin{ShortVersionBlock}
\section{Technical Backgrounds}\label{sec:preliminaries}

We denote the reals, integers, naturals and positive naturals by $\R$, $\Z$, $\N$, and $\Zp$, respectively.
For $\real \in \R$, we let $\lfloor \real \rfloor \in \Z$ be the maximum integer satisfying $\lfloor \real \rfloor \leq \real$.
For a set $X$, we denote its powerset by $\powerset{X}$,
the set of infinite sequences of $X$ by $X^{\omega}$, and
the set of sequences of $X$ of length $n$ by $X^{n}$, where $n \in \N$.
We let $X^{*} = \bigcup_{n \in \N} X^{n}$ and $X^{\infty} = X^{*} \cup X^{\omega}$.
For $\word \in X^{*}$ and $\word' \in X^{\infty}$, we let $\word \cdot \word' \in X^{\infty}$ be their juxtaposition.
For $\word \in X^{*}$, we let $\prefix{\word}{n}$ be the prefix of $\word$ of length $n$ for $n \leq |\word|$,
where $|\word|$ is the length of $\word$.
We let $\emptyword$ be the empty sequence.
For a DFA $\A$, we denote its language by $\Lg(\A)$.

\subsection{Discrete-time signal temporal logic}\label{subsec:STL}
Let $\Var$ be the finite set of variables.
A (discrete-time) \emph{signal} $\signal \in {(\R^{\Var})}^{\infty}$ is a finite or infinite sequence of functions $\valuation_i\colon\Var \to \R$.
We\LongVersion{ also} call such $\valuation_i$ a (signal) valuation.

\emph{Signal temporal logic (STL)}~\cite{DBLP:conf/formats/MalerN04} is a widely used formalism to represent behaviors of signals.
We use its variant for discrete-time signals.

\begin{definition}
 [signal temporal logic]%
 \label{def:STL}
 For a finite set $\Var$ of variables, the syntax of \emph{signal temporal logic (STL)} is defined as follows, %
where\footnote{In our implementation and experiments, we extend $\fun$ to receive a bounded history of signal valuations, \ie{} $\fun\colon {(\R^{\Var})}^N \to \R$, where $N \in \Zp$ is the history bound.} $\fun\colon \R^{\Var} \to \R$, $\constant \in \R$, and $i,j \in \N \cup \{+\infty\}$ satisfying $i < j$.
\[
 \fml, \fml' \Coloneq \top \mid \fun \geq \constant \mid \neg \fml \mid \fml \lor \fml' \mid \NextOp \fml \mid \fml \UntilOp{[i,j)} \fml'
\]
For an STL formula $\fml$, we let $\Pred(\fml)$ be the set of inequalities $\fun \geq \constant$ in $\fml$.
We define the following as syntax sugar:
$\bot \equiv \neg \top$,
$\fun < \constant \equiv \neg (\fun \geq \constant)$,
$\fml \imply \fml' \equiv (\neg \fml) \lor \fml'$,
$\fml \land \fml' \equiv \neg (\neg \fml \lor \neg \fml')$,
$\DiaOp{[i,j)} \fml \equiv \top \UntilOp{[i,j)} \fml$,
$\BoxOp{[i,j)} \fml \equiv \neg \DiaOp{[i,j)} \neg \fml$,
$\fml \Release{[i,j)} \fml' \equiv \neg (\neg \fml \UntilOp{[i,j)} \neg \fml')$, and
$\fml \TRelease{[i,j)} \fml' \equiv \fml \Release{[i, j)} (\fml \lor \fml')$.
\end{definition}
For an STL formula $\fml$, an \emph{infinite} signal $\signalWithInfiniteInside$, and $k \in \N$,
the satisfaction relation $\satisfy{\signal}{k}{\fml}$ is inductively defined as follows.
\begin{gather*}
 \satisfy{\signal}{k}{\top} \qquad
 \satisfy{\signal}{k}{\fun \geq \constant} \iff \fun(\signal_k) \geq \constant \qquad
 \satisfy{\signal}{k}{\neg \fml} \iff (\signal, k) \not\models \fml \\
 \satisfy{\signal}{k}{\fml \lor \fml'} \iff \satisfy{\signal}{k}{\fml} \lor \satisfy{\signal}{k}{\fml'} \quad
 \satisfy{\signal}{k}{\NextOp{\fml}} \iff \satisfy{\signal}{k+1}{\fml}\\
 \begin{aligned}
  \satisfy{\signal}{k}{\fml \UntilOp{[i,j)} \fml'} &\iff
  \exists l \in \{k+i,k+i+1,\dots,k+j-1\}.\,\satisfy{\signal}{l}{\fml'} \\
  &\qquad\quad\land \forall m \in \{k,k+1,\dots,l-1\}.\, \satisfy{\signal}{m}{\fml}
 \end{aligned}
\end{gather*}

We let $\signal \models \fml$ if we have $\satisfy{\signal}{0}{\fml}$.
An STL formula $\fml$ is\LongVersion{ called} a \emph{safety} property if
for any infinite signal $\signal$
satisfying $\signal \not\models \fml$,
there is a finite prefix $\signal_f$ of $\signal$ such that
for any infinite signal $\signal'$, we have
$\signal_f \cdot \signal' \not\models \fml$.
Such a prefix is called \emph{bad}~\cite{DBLP:journals/fmsd/KupfermanV01}.
Since a discrete-time STL formula $\fml$ is easily representable by an LTL~\cite{DBLP:conf/focs/Pnueli77} formula with $\Pred(\fml)$ as atomic propositions,
violation of a safety STL formula $\fml$ can be monitored by a DFA $\Monitor_{\fml}$ over $\powerset{\Pred(\fml)}$.
Namely, the DFA $\Monitor_{\fml}$ takes a sequence $a_1, a_2, \dots a_n \in {(\powerset{\Pred(\fml)})}^*$
and decides if the corresponding prefix $\signal_f$ of the signal under monitoring is bad or not for $\fml$,
where $a_i = \{\fun \geq \constant \in \Pred(\fml) \mid \fun(\signal_i) \geq \constant\}$.
See, \eg{}~\cite{DBLP:journals/tosem/BauerLS11}, for a construction of such\LongVersion{ a DFA} $\Monitor_{\fml}$ (for LTL).

\subsection{Fully Homomorphic Encryption}\label{subsec:FHE}
\emph{Homomorphic encryption (HE)}~\cite{rivest1978data} is a kind of encryption that enables data evaluation without decryption, \ie{}
for a plaintext $\plain$ and a function $f$,
$f(\plain)$ is (nearly) equal to
$\dec(f_{\mathrm{HE}}(\enc(\plain)))$, where $\enc(x)$ and $\dec(y)$ are the encryption and decryption results
and $f_{\mathrm{HE}}$ is the HE counterpart of $f$.
\emph{Fully HE (FHE)}~\cite{DBLP:conf/stoc/Gentry09} is a kind of HE such that any function $f$ can be evaluated without decryption.

In FHEs, ciphertexts based on the \emph{learning with error (LWE)} problem~\cite{DBLP:journals/jacm/Regev09} and its ring variant \emph{(RLWE)}~\cite{DBLP:journals/jacm/LyubashevskyPR13} are widely used.
Let $n, \quotient \in \Zp$ be security parameters.
An LWE ciphertext $\cipher_{\quotient}$ represents a value in the quotient ring $\qRing$ of $\Z$ modulo $\quotient$.
An RLWE ciphertext $\RLWE{\cipher}_{\quotient, N}$ represents a polynomial of degree $N - 1$ with coefficients in $\qRing$.
We omit $\quotient$ and $N$ if they are clear from the context.
For both LWE and RLWE, a public key $\publicKey$ and a private key $\key$ can be used for encryption while decryption requires the private key.

For security purpose,
small \emph{noise} is added to LWE and RLWE ciphertexts,
and encryption and decryption slightly change the value, \ie{}
for a plaintext $\plain$, $\dec(\enc(\plain))$ is slightly different from $\plain$.
Each FHE scheme has its approach to handle the noise, which we will review later.
However, the noise increases over FHE operations, and eventually, the result of the decryption largely deviates from the expected value, which is considered as a failure.
\LongVersion{For example, for LWE ciphertexts, homomorphic addition
doubles the noise on average.}
The following two approaches are widely used to address this issue:
\begin{ienumeration}
 \item Assuming the number of applications of FHE operations, the noise is sampled so that the decryption is successful throughout the computation;
 \item The noise is reduced by a special FHE operation called \emph{bootstrapping}~\cite{DBLP:conf/stoc/Gentry09} before it becomes too large.
\end{ienumeration}
\begin{table}[tbp]
 \centering
 \scriptsize
 \caption{Comparison of the TFHE and CKKS schemes.}\label{table:TFHE_vs_CKKS}
 \begin{tabular}{lccc}
  \toprule
  & Supported Operations & Our usage & Noise reduction by bootstrapping \\
  \midrule
  CKKS & Polynomial operations & Arithmetic operations & Slow \\
  TFHE & Any (by LUTs) & Logical operations & Relatively fast \\
\bottomrule
\end{tabular}
\end{table}
\begin{table*}[tbp]
 \centering
 \scriptsize
 \caption{Summary of the ciphertexts used in our protocol.}\label{table:ciphertexts}
 \begin{tabular}{lcccccc}
  \toprule
  & Encoded values & Encoded values & Example usage in & Notation in \\
  & in CKKS & in TFHE & our protocol & this paper\\
  \midrule
  LWE & --- & Boolean value & Monitoring results & bold font, \eg{} $\cipher$ \\
  RLWE & (Approx.) Reals & Boolean array & Signal valuations & bold font + overline, \eg{} $\RLWE{\cipher}$ \\
  RGSW & --- & Boolean value & Truth values of predicates & bold font + tilde, \eg{} $\RGSW{\cipher}$ \\
  \bottomrule
 \end{tabular}
\end{table*}

Among various FHE schemes, \emph{Cheon-Kim-Kim-Song scheme (CKKS)}~\cite{DBLP:conf/asiacrypt/CheonKKS17} and \emph{FHE over the Torus (TFHE)}~\cite{DBLP:journals/joc/ChillottiGGI20} are two of the most widely used schemes.
\cref{table:TFHE_vs_CKKS} summarizes their comparison.
\cref{table:ciphertexts} summarizes the encoded value for each kind of ciphertexts in each scheme and our specific usage.

In CKKS, RLWE ciphertexts are usually used.
CKKS is typically used to approximately encode real values and apply polynomial operations, \eg{} addition, subtraction, and multiplication.
When representing a real value $\real \in \R$ by an RLWE ciphertext\footnote{For simplicity, we omit the embedding of multiple values as roots of a polynomial~\cite{DBLP:conf/asiacrypt/CheonKKS17} for vectorization, which we do not use in our implementation.},
we first approximate $\real$ by a multiplication of $\plain \in \qRing$ and a positive floating-point number $\scale$, where $\plain$ represents a bit-encoding of a signed integer.
Then, we use a pair $(\scale, \RLWE{\cipher}_{q, N})$ of $\scale$ and an RLWE ciphertext $\RLWE{\cipher}_{q, N}$ encrypting a polynomial $\RLWE{\plain}$ with constant term $\plain$ to represent $\real$.
Notice that $\scale$ is not encrypted because it is typically chosen independently of $\real$.
We usually omit $\scale$ and simply write $\RLWE{\cipher}_{q, N}$.
We let $\encCKKS$ and $\decCKKS$ be the encryption and decryption with the above encoding.
In CKKS, errors are simply ignored, assuming that they are much smaller than, \eg{} noise in sampling and numeric error.
While CKKS is efficient for polynomial operations, it does not support non-polynomial operations.
Bootstrapping in CKKS is known to be slow~\cite{DBLP:conf/eurocrypt/CheonHKKS18}, and thus, CKKS is typically employed in situations where an upper bound of the number of applications of FHE operations is known beforehand.
In TFHE, both LWE and RLWE ciphertexts are used.
TFHE supports a multiplexer operation: an FHE operation to select one of the values of the given ciphertexts according to the value of another ciphertext (called a control bit).
By combining multiplexers, one can implement a look-up table (LUT), which allows to encode any operations with TFHE.\@
In this paper, we use TFHE to encode Boolean values and logical operations, such as AND and OR.\@
Specifically, we use LWEs to represent Boolean values and RLWEs to represent Boolean arrays.
One typical encoding used for LWEs is such that $\plain \in \qRing$ represents $\top$ if and only if $\plain \in \{0, 1, \dots, \quotient / 2 - 1\}$.
The encoding for RLWEs is similar: we use coefficients as an array of values.
The result of the decryption (as a Boolean value) is successful if
the noise is small\LongVersion{ enough,} and the plaintext is in the expected range.
We let $\encTFHE$ and $\decTFHE$ be the encryption and decryption with the above encoding.
Although TFHE supports any operations, it is not as fast as CKKS for polynomial operations, \eg{} a single multiplication of two 16-bit integer values takes more than a few seconds~\cite{DBLP:journals/joc/ChillottiGGI20}.
Bootstrapping is relatively fast in TFHE, and\LongVersion{ thus,} TFHE is suitable for handling unbounded length of data, reducing the noise by bootstrapping.
Another type of ciphertexts called \emph{RGSW}~\cite{DBLP:conf/crypto/GentrySW13} are also used in TFHE.\@
An RGSW ciphertext is, roughly speaking, a collection of RLWE ciphertexts.
When conducting logical operations\LongVersion{ with TFHE}, an RGSW ciphertext represents a Boolean value that is used, \eg{} as the control bit of multiplexers.

\subsection{Online algorithm for oblivious DFA execution}\label{subsec:HomFA}
Banno et al.~\cite{DBLP:conf/cav/BannoMMBWS22} proposed two TFHE-based algorithms (\REVERSE{} and \BLOCK{}) to obliviously execute a DFA over $\{\top, \bot\}$.
They use these algorithms for \emph{online} LTL monitoring by first constructing a DFA $\A$ over $\{\top, \bot\}$ from an LTL formula $\fml$ over atomic propositions $\AP$.
The construction is by
\begin{ienumeration}
 \item making a DFA over $\powerset{\AP}$ from $\fml$, \eg{} with~\cite{DBLP:journals/tosem/BauerLS11}, and
 \item modifying its alphabet to $\{\top, \bot\}$ by encoding each $\action^{\AP} \in \powerset{\AP}$ by $\action_1, \action_2,\dots,\action_{|\AP|} \in {\{\top, \bot\}}^{|\AP|}$.
\end{ienumeration}

Given a sequence $\encWordSequence$ of RGSW ciphertexts representing the input word $\wordWithInside \in {\{\top, \bot\}}^*$,
their algorithms return a sequence of LWE ciphertexts representing if prefixes of the input word are accepted by $\A$.
Their algorithms incrementally process the input word and return prefixes of the result before obtaining the entire input.
Thus, they are suitable for online monitoring.

There algorithms are based on another TFHE-based algorithm in~\cite{DBLP:conf/cav/BannoMMBWS22} for \emph{offline} execution of DFAs.
The offline algorithm reads the input word $\word$ \emph{from the tail} to achieve DFA execution only using multiplexers, which is fast for TFHE.\@
In \REVERSE{}, the offline algorithm is applied to the reversed DFA $\reverse{\A}$, \ie{} the minimum DFA such that $\word \in \Lg(\A) \iff \reverse{\word} \in \Lg(\reverse{\A})$, where $\reverse{\word}$ is the reversed word of $\word$.
\REVERSE{} returns the result for each prefix $\prefix{\word}{i}$ of $\word$.

In \BLOCK{}, they use a variant of the offline algorithm that also returns an RLWE ciphertext representing the state after reading the given ciphertexts.
The input word $\word$ is split into a series of words $\word_1, \word_2, \dots, \word_N \in {({\{\top, \bot\}}^{\BlockSize})}^{*}$ of length $\BlockSize \in \Zp$,
each word $\word_j$ is fed to the modified offline algorithm from $j = 1$ to $j = N$,
and the result for each block is used to choose the state after reading $\word_1, \word_2, \dots, \word_j$.
Since \BLOCK{} feeds each $\word_j \in {\{\top, \bot\}}^{\BlockSize}$ to the offline algorithm,
it returns the result for prefixes $\prefix{\word}{i}$ such that $i = \BlockSize \times j$ for some $j \in \N$.

While \REVERSE{} has better scalability with respect to the size of the original LTL formula in terms of the worst case complexity,
\cite{DBLP:conf/cav/BannoMMBWS22} reports that their practical superiority is inconclusive:
\REVERSE{} is typically faster than \BLOCK{} but \REVERSE{} takes much longer time, \eg{} when the reversed DFA $\reverse{\A}$ is huge.

\end{ShortVersionBlock}

\begin{ShortVersionBlock}
\section{Scheme switching optimized with value range information}\label{appendix:section:protocol:ckks_to_tfhe}
 \SetKwFunction{FSampleExtract}{RLWEToLWE}
 \SetKwFunction{FCircuitBootstrapping}{LWEToRGSW}
 \SetKwFunction{FModSwitch}{modSwitch}
 \SetKwFunction{FPartialHomDecomp}{PartHomDecomp}
 \begin{algorithm}[tbp]
 \caption{Outline of \texttt{isPositive}.}%
 \label{algorithm:protocol:ckks_to_tfhe}
 \DontPrintSemicolon{}
 \scriptsize
 \newcommand{\myCommentFont}[1]{\texttt{\scriptsize{#1}}}
 \SetCommentSty{myCommentFont}
 \KwIn{An RLWE ciphertext $\diffCKKS$, its scaling factor $\scale \in \R$ for CKKS, the maximum value $\diffMax \in \Rp$ of $|\decCKKS(\diffCKKS)|$, the switching size $k \in \Zp$ in $\FHomDecomp$, and the modulus $\quotientCKKS$ and $\quotientTFHE$ of\LongVersion{ the quotient ring in} the input and output ciphertexts}
 \KwOut{An RGSW ciphertext $\signTFHE[\RGSW]$ such that if $\signTFHE[\RGSW]$ and $\diffCKKS$ are decrypted correctly, $\decTFHE(\signTFHE[\RGSW]) = \top \iff \decCKKS(\diffCKKS) \geq 0$ holds}
 $\diffScaledCKKS \gets \FEvalCKKS(\diffCKKS \times \frac{\lfloor \quotientCKKS / 2 - 1\rfloor}{\diffMax \times \scale})$\tcp*[r]{Scale the value space keeping the sign}\label{algorithm:protocol:ckks_to_tfhe:amplify}
 $\diffSwitched[\RLWE] \gets \FModSwitch(\diffScaledCKKS, \quotientCKKS \to \quotient_{\mathrm{tmp}})$\tcp*[r]{Switch to the intermediate modulus $\quotient_{\mathrm{tmp}}$}\label{algorithm:protocol:ckks_to_tfhe:rescaling}
 $\diffSwitched \gets \FSampleExtract(\diffSwitched[\RLWE])$\label{algorithm:protocol:ckks_to_tfhe:sample_extract}\;
 $\signTFHE[] \gets \FHomDecomp(\diffSwitched, k)$\tcp*[r]{Switch the security parameters\LongVersion{ for TFHE} preserving the MSB}\label{algorithm:protocol:ckks_to_tfhe:key_switching}
 $\signTFHE \gets \FCircuitBootstrapping(\signTFHE[])$\label{algorithm:protocol:ckks_to_tfhe:circuit_bootstrapping}\;
 \end{algorithm}

 \MW{To Kotaro: Please read this section}
 \LongVersion{Among the steps in \cref{algorithm:protocol:outline},
 the most non-trivial and computationally demanding one is \texttt{isPositive} at \cref{algorithm:protocol:outline:ckks-to-tfhe},
 the construction of an RGSW ciphertext in TFHE.\@
 We elaborate it here.}

 As we explained in \cref{subsec:FHE},
 when we encrypt $\real \in \R$ within CKKS,
 we use $\plain \in \qRing$ satisfying %
 $\real \geq 0$ if and only if $\plain \in \{0, 1, \dots, \lfloor \quotient / 2 - 1\rfloor\}$.
 Since this coincides with the encoding of $\top$ within TFHE if security parameters are appropriately switched,
 we can obtain an RGSW ciphertext $\signTFHE$ satisfying $\decTFHE(\signTFHE) = \top \iff \fun(\signal_i) \geq \constant$ by
 converting the given RLWE ciphertext to an RGSW ciphertext with appropriate switching of security parameters.
 In this construction, switching of security parameters is the most computationally demanding
 even if we use a recent FHE operation called \emph{homomorphic decomposition} ($\FHomDecomp$)~\cite{DBLP:journals/tches/MaHWZW24}, which is more efficient than the operation used in~\cite{DBLP:conf/ccs/00010PMZJG23}.
 Thus, we optimize this step to improve the efficiency of the overall algorithm.
 Intuitively, $\FHomDecomp$ slices the binary representation of an integer into multiple short bit sequences.
 Given an LWE ciphertext $\cipher_{\quotient}$ with modulus $\quotient$ and $W, k \in \Zp$,
 $\FHomDecomp$ decomposes the first $W \times k$ bits of $\cipher_{\quotient}$ into LWE ciphertexts $\cipher_{\quotient', 1}, \cipher_{\quotient', 2}, \dots, \cipher_{\quotient', k}$ with modulus $\quotient'$ such that
 the first $W$ bits of each $\cipher_{\quotient', i}$ represents the ``$i$-th block'' of $\cipher_{\quotient}$, where each block consists of $W$ bits.
 The decomposed ciphertexts may have different security parameters from $\cipher_{\quotient}$.

 Although the result of decryption as a Boolean value solely depends on the ``most significant bit (MSB)'' of $\plain$,
 naively, we need to decompose all the blocks of $\cipher_{\quotient}$.
 This is because if $\plain$ is near the boundaries of flipping (\ie{} $0$, $\lfloor \quotient / 2 - 1 \rfloor$, $\lfloor \quotient / 2 - 1 \rfloor + 1$, and $\quotient - 1$),
 the higher blocks are even closer to the boundaries and quite fragile to the noise by scheme switching.
 Nevertheless, if we have the possible range of the encrypted raw value $\plain$ as domain knowledge,
 we can scale $\plain$ so that it is sufficiently far from the boundaries.
 The range of $\plain$ can be estimated from the range of each $\var \in \Var$ in the monitored signal, $\scale$ in CKKS, and the monitored formula $\fml$.
 With such preprocessing, we apply $\FHomDecomp$ only partly, \eg{} to the higher half of the blocks.

 \cref{algorithm:protocol:ckks_to_tfhe} outlines \texttt{isPositive}.
 In \cref{algorithm:protocol:ckks_to_tfhe:amplify}, we scale $\diffCKKS$ so that the range of $\decCKKS(\diffCKKS)$ almost matches the range\LongVersion{ $[- \lfloor \quotientCKKS / 2\rfloor / \scale, \lfloor \quotientCKKS / 2 - 1\rfloor / \scale]$} of the ciphertext.
 In \cref{algorithm:protocol:ckks_to_tfhe:rescaling}, we switch the modulus from $\quotientCKKS$ to $\quotientTFHE$.
 In \cref{algorithm:protocol:ckks_to_tfhe:sample_extract}, we use an FHE operation called \emph{sample extraction}~\cite{DBLP:journals/joc/ChillottiGGI20}
 to obtain an LWE ciphertext $\diffSwitched$ satisfying $\decTFHE(\diffSwitched) = \top \iff \decCKKS(\diffSwitched[\RLWE]) \geq 0$ if the decryption is successful.
 In \cref{algorithm:protocol:ckks_to_tfhe:key_switching}, we apply homomorphic decomposition to obtain another LWE ciphertext $\signTFHE[]$ with different security parameters, preserving the decryption result if it is successful.
 Lastly, in \cref{algorithm:protocol:ckks_to_tfhe:circuit_bootstrapping}, we use another FHE operation called \emph{circuit bootstrapping}~\cite{DBLP:journals/joc/ChillottiGGI20} to translate $\signTFHE[]$ to an RGSW ciphertext $\signTFHE$, preserving the decryption result if it is successful. %
 Notice that even if we do not have $\diffMax$, we can still use this algorithm by 
 \begin{ienumeration}
 \item skipping \cref{algorithm:protocol:ckks_to_tfhe:amplify} and
 \item applying $\FHomDecomp$ to all the bits, \ie{} $k = \quotient_{\mathrm{tmp}}$ in \cref{algorithm:protocol:ckks_to_tfhe:key_switching}.
 \end{ienumeration}

\end{ShortVersionBlock}

\IEEEVersion{\subsection}\LNCSVersion{\section}{Detail of the FHE parameters}\label{section:parameters}
\begin{table}[tb]
 \centering
 \caption{Summary of CKKS parameters}%
 \label{table:ckks_params}
 \scriptsize
 \begin{tabular}{c c c}
  \toprule
  Parameter& Value we used & Description \\
  \midrule
  $\mathit{poly\_modulus\_degree}$ & 8192 & \begin{tabular}{c}
                                             The degree of\\
                                             polynomial modulus
                                            \end{tabular}\\
  $\mathit{base\_sizes}$  & $60, 40, 40, 60$ & The bit size of each level\\
  $\scale$ & $2^{40}$ & The scale of fresh ciphertexts\\
  \bottomrule
 \end{tabular}
\end{table}
\begin{table}[tb]
 \centering
 \caption{Summary of TFHE parameters}%
 \label{table:tfhe_params}
 \scriptsize
 \begin{tabular}{c c c}
  \toprule
  Parameter& Value we used & Description \\
  \midrule
  $q$ & $2^{32}$ & The modulus for level 0, 0.5 and 1 ciphertexts\\
  $\overline{q}$ & $2^{64}$ & The modulus for level 2 and 3 ciphertexts\\
  $N_0$& $635$ & The length of level 0 ciphertexts\\
  $\alpha_0$ & $2^{-15}$ & \begin{tabular}{c}The standard deviation of the noise\\ for fresh level 0 ciphertexts\end{tabular}\\
  $N_1$ & $2^{10}$ & \begin{tabular}{c}
                    The length of level 1 ciphertexts \\
                    and the dimension of level 1 ciphertexts\\
                   \end{tabular}\\
  $\alpha_1$ & $2^{-25}$ & \begin{tabular}{c}The standard deviation of the noise\\ for fresh level 1 ciphertexts\end{tabular}\\
  $N_2$& $2^{11}$ & \begin{tabular}{c}
                              The length of level 2 ciphertexts \\
                              and the dimension of level 2 ciphertexts\\
                             \end{tabular}\\
  $\alpha_{2}$ & $2^{-47}$ & \begin{tabular}{c}The standard deviation of the noise\\ for fresh level 2 ciphertexts\end{tabular}\\
  $N_{0.5}$& $760$ & The length of level 0.5 ciphertexts\\
  $\alpha_{0}$ & $2^{-17}$ & \begin{tabular}{c}The standard deviation of the noise\\ for fresh level 0.5 ciphertexts\end{tabular}\\
  $N_{3}$& $2^{13}$ & \begin{tabular}{c}
                              The length of level 3 ciphertexts \\
                              and the dimension of level 3 ciphertexts\\
                             \end{tabular}\\
  $\alpha_{3}$ & $2^{-47}$ & \begin{tabular}{c}The standard deviation of the noise\\ for fresh level 3 ciphertexts\end{tabular}\\
  $l$ & 3 & \begin{tabular}{c}Half of the number of rows in\\ level 1 RGSW ciphertexts\end{tabular}  \\
  $\mathit{Bg}_1$ & $2^{6}$ & The base for multiplexers with level 1\\
  $\overline{l}$ & 4 & \begin{tabular}{c}Half of the number of rows in\\ level 2 RGSW ciphertexts\end{tabular}  \\
  $\mathit{Bg}_2$ & $2^{9}$ & The base for multiplexers with level 2\\
  \bottomrule
 \end{tabular}
\end{table}

For the security parameters of CKKS,
we used the default parameters of SEAL~\cite{sealcrypto}, which satisfy the 128-bit security level.
\cref{table:ckks_params} shows the parameters of SEAL that users must provide.

For the security parameters of TFHE,
we also used the default parameters of TFHEpp~\cite{DBLP:conf/uss/MatsuokaBMS021}.
These parameters have been confirmed to satisfy a 128-bit security level by a lattice-estimator.
In addition to the default security parameters of TFHEpp (\ie{} levels 0, 1, and 2),
we also used two other security parameters, levels 0.5 and 3, for scheme switching.
\cref{table:tfhe_params} lists some of the concrete parameters.

\IEEEVersion{\subsection}\LNCSVersion{\section}{Detail of the benchmarks}
\IEEEVersion{\subsubsection}\LNCSVersion{\subsection}{Detail of the \BloodGlucose{} benchmark}

The monitored log in the \BloodGlucose{} benchmark is generated by running a simulator simglucose~\cite{simglucose} and sampling the blood glucose level at a fixed sampling rate.
We used the sampling rate \SI{1}{\minute}, \ie{}
for a signal $\signalWithFiniteInside$, the difference of the time point of $\signal_i$ and $\signal_{i+1}$ is \SI{1}{\minute}.

\IEEEVersion{\subsubsection}\LNCSVersion{\subsection}{Detail of the \RSS{} benchmark}\label{section:detail_rss_formula}

The monitored log in the \RSS{} benchmark is generated by running a driving scenario in an unpublished 2D driving simulator and sampling the behavior at a fixed sampling rate.
We used the sampling rate \SI{500}{\milli\second}, \ie{}
for a signal $\signalWithFiniteInside$, the difference of the time point of $\signal_i$ and $\signal_{i+1}$ is \SI{500}{\milli\second}.

\cref{table:parameters_rss_formula} shows the parameters used in the STL formula for the RSS rule.

\begin{table}[tbp]
 \centering
 \caption{Parameters in the STL formula for RSS.}\label{table:parameters_rss_formula}
 \begin{tabular}{llr}
  \toprule
  & Description & Value we used\\
  \midrule
  $w$ & width of a lane & \SI{4}{\metre} \\
  $\reactionDelay$ & delay of the reaction & \SI{1}{\second}\\
  $\reactionDelayInStep$ & delay of the reaction / sampling rate & 2 samples\\
  $\ayMaxAcc$ & maximum of allowed acceleration & \SI{2}{\metre\per\square\second}\\
  $\ayMaxBr$ & maximum of allowed brake & \SI{9}{\metre\per\square\second}\\
  $\ayMinBr$ & \begin{tabular}{c}
                minimum brake required \\
                at a critical situation
               \end{tabular} 
  & \SI{7}{\metre\per\square\second}\\
  \bottomrule
 \end{tabular}
\end{table}
\IEEEVersion{\subsection}\LNCSVersion{\section}{Additional experimental results}\label{section:additional_experiment_result}

To observe the memory requirements on the clients, we also measured the memory consumption for encryption and decryption on SBCs.
We measured the maximum resident set size using GNU Time\footnote{GNU Time is available from \url{https://www.gnu.org/software/time/}.}.
We measured it for five times and report the average.

\begin{table}[tbp]
 \centering
 \caption{Mean memory consumption to encrypt RLWE ciphertexts within the CKKS scheme and decrypt LWE ciphertexts within the TFHE scheme on SBCs.}%
 \label{table:enc_dec_mem_compsumption_result}
 \begin{tabular}{c c c c}
  \toprule
  & Enc.\ w/ public key & Enc.\ w/ private key & Dec. \\
  \midrule
  NanoPi R6S (\emph{w/} AES accelerator) & \SI{360656}{\kilo\byte} & \SI{298951.2}{\kilo\byte} & \SI{6876.8}{\kilo\byte} \\
  Raspberry Pi 4 (\emph{w/o} AES accelerator) & \SI{360704}{\kilo\byte} & \SI{299089.6}{\kilo\byte} & \SI{7168}{\kilo\byte} \\
  \bottomrule
 \end{tabular}
\end{table}

\cref{table:enc_dec_mem_compsumption_result} shows the mean memory consumption.
The encryption requires significantly more memory than the decryption.
This is likely because of the difference in the size of the keys between CKKS and TFHE.
Nevertheless, the memory consumption is less than \SI{400}{\mega\byte}.
This is likely small enough for most of the recent IoT devices, including the SBCs we used.

\IEEEVersion{\subsection}\LNCSVersion{\section}{Detailed experimental results}\label{section:detail_experiment_result}
\begin{table*}[tbp]
 \rotatebox{90}{
 \begin{minipage}{\textheight}
 \centering
 \caption{Breakout of the execution time of \ourTool{}. The cells with \TIMEOUT{} are the configurations such that the DFA construction aborted due to out of memory.}\label{table:result_detail}
 \scriptsize
  \begin{tabular}{lllllllllllllllll}
   \toprule
   & \multicolumn{4}{c}{DFA eval. (sec.)} & \multicolumn{4}{c}{CKKS to TFHE (sec.)} & \multicolumn{4}{c}{CKKS eval. (sec.)} & \multicolumn{4}{c}{Runtime (sec.)} \\
   & \multicolumn{2}{c}{\Optimized{}} & \multicolumn{2}{c}{\Naive{}} & \multicolumn{2}{c}{\Optimized{}} & \multicolumn{2}{c}{\Naive{}} & \multicolumn{2}{c}{\Optimized{}} & \multicolumn{2}{c}{\Naive{}} & \multicolumn{2}{c}{\Optimized{}} & \multicolumn{2}{c}{\Naive{}} \\
   & \BLOCK{} & \REVERSE{} & \BLOCK{} & \REVERSE{} & \BLOCK{} & \REVERSE{} & \BLOCK{} & \REVERSE{} & \BLOCK{} & \REVERSE{} & \BLOCK{} & \REVERSE{} & \BLOCK{} & \REVERSE{} & \BLOCK{} & \REVERSE{} \\
   \midrule
   $\BloodGlucose_{1}$ & 6.09e+01 & 2.99e+02 & 6.10e+01 & 3.02e+02 & 1.72e+02 & 1.76e+02 & 2.99e+02 & 3.12e+02 & 5.77e-02 & 5.55e-02 & 5.51e-02 & 5.32e-02 & 2.33e+02 & 4.75e+02 & 3.60e+02 & 6.13e+02 \\
   $\BloodGlucose_{2}$ & 6.02e+01 & 3.01e+02 & 6.48e+01 & 3.00e+02 & 1.69e+02 & 1.76e+02 & 3.08e+02 & 3.18e+02 & 7.29e-02 & 7.01e-02 & 7.30e-02 & 6.99e-02 & 2.29e+02 & 4.76e+02 & 3.73e+02 & 6.18e+02 \\
   $\BloodGlucose_{4}$ & 1.23e+01 & \TIMEOUT{} & 1.33e+01 & \TIMEOUT{} & 1.70e+02 & \TIMEOUT{} & 3.15e+02 & \TIMEOUT{} & 8.59e-02 & \TIMEOUT{} & 8.67e-02 & \TIMEOUT{} & 1.83e+02 & \TIMEOUT{} & 3.28e+02 & \TIMEOUT{} \\
   $\BloodGlucose_{5}$ & 1.57e+02 & \TIMEOUT{} & 1.56e+02 & \TIMEOUT{} & 1.73e+02 & \TIMEOUT{} & 3.02e+02 & \TIMEOUT{} & 5.67e-02 & \TIMEOUT{} & 5.34e-02 & \TIMEOUT{} & 3.30e+02 & \TIMEOUT{} & 4.59e+02 & \TIMEOUT{} \\
   $\BloodGlucose_{6}$ & 1.87e+02 & \TIMEOUT{} & 1.59e+02 & \TIMEOUT{} & 1.87e+02 & \TIMEOUT{} & 3.09e+02 & \TIMEOUT{} & 8.22e-02 & \TIMEOUT{} & 7.13e-02 & \TIMEOUT{} & 3.74e+02 & \TIMEOUT{} & 4.68e+02 & \TIMEOUT{} \\
   \midrule{}
   $\BloodGlucose_{7}$ & 9.79e+02 & 1.18e+01 & 9.56e+02 & 1.13e+01 & 2.98e+03 & 2.91e+03 & 4.94e+03 & 5.04e+03 & 1.65e+00 & 1.95e+00 & 1.62e+00 & 1.91e+00 & 3.96e+03 & 2.93e+03 & 5.90e+03 & 5.06e+03 \\
   $\BloodGlucose_{8}$ & 9.63e+02 & 1.24e+01 & 9.67e+02 & 1.28e+01 & 2.91e+03 & 3.01e+03 & 5.08e+03 & 5.07e+03 & 2.49e+00 & 3.17e+00 & 2.60e+00 & 3.10e+00 & 3.87e+03 & 3.03e+03 & 6.05e+03 & 5.09e+03 \\
   $\BloodGlucose_{10}$ & 1.15e+03 & 1.23e+01 & 1.14e+03 & 1.16e+01 & 2.51e+03 & 2.48e+03 & 4.29e+03 & 4.25e+03 & 8.26e-01 & 8.87e-01 & 8.11e-01 & 8.39e-01 & 3.66e+03 & 2.50e+03 & 5.43e+03 & 4.26e+03 \\
   $\BloodGlucose_{11}$ & 1.09e+03 & 1.25e+01 & 1.14e+03 & 1.13e+01 & 2.40e+03 & 2.63e+03 & 4.26e+03 & 4.22e+03 & 7.79e-01 & 9.81e-01 & 7.84e-01 & 8.49e-01 & 3.49e+03 & 2.65e+03 & 5.40e+03 & 4.23e+03 \\
   \midrule{}
   \RSS{} & 4.89e+00 & 5.11e-01 & 5.10e+00 & 4.94e-01 & 2.25e+01 & 2.41e+01 & 4.07e+01 & 3.77e+01 & 4.61e-01 & 4.77e-01 & 4.79e-01 & 4.60e-01 & 2.79e+01 & 2.50e+01 & 4.62e+01 & 3.87e+01 \\
   \bottomrule
  \end{tabular}
 \end{minipage}}
\end{table*}

\cref{table:result_detail} shows the breakout of the execution time of \ourTool{}.
This shows that scheme switching is indeed the major bottleneck among the entire process.

\ifdefined\VersionWithComments%
\IEEEVersion{\subsection}\LNCSVersion{\section}{Additional information for us, the authors}
\instructions{(RV'24) Regular papers (up to 16 pages, not including references) should present original unpublished results. We welcome theoretical papers, system design papers, papers describing domain-specific variants of RV, and real-world case studies of runtime verification. The page limitations mentioned above include all text and figures, but exclude references. Authors can include a clearly marked appendix exceeding the 16 pages, but this will be reviewed at the discretion of reviewers, and will not be included in the proceedings. This appendix must follow the References and be no longer than 4 pages.}
\begin{lncs}
 \setcounter{tocdepth}{1}
\end{lncs}
\listoftodos{}
\fi
\end{LongVersionBlock}
\end{document}